\documentclass[11pt]{article}
\usepackage{lineno,hyperref}
\modulolinenumbers[5]

\usepackage{amsmath}
\usepackage{amssymb}
\usepackage{amsthm}
\usepackage{tikz}
\usepackage{pgfplots}
\usepackage{graphicx}
\usepackage{slashbox}
\usepackage{caption}
\usepackage{subcaption}
\usepackage{xcolor}
\allowdisplaybreaks  
\usepackage{cite}
\usepackage{algorithm}
\usepackage[noend]{algcompatible}
\usepackage{mathtools}
\usepackage{subcaption}
\usepackage{hyperref}
\usepackage{cleveref}
\usepackage{multirow}

\allowdisplaybreaks 

\newtheorem{assumption}{Assumption} 
\newtheorem{problem}{Problem} 
\newtheorem{remark}{Remark}
\newtheorem{lemma}{Lemma}

\newtheorem{corollary}{Corollary}
\newtheorem{example}{Example}
\usepackage{authblk}
\usepackage{comment}

\newtheorem{defi}{Definition}
\newtheorem{theo}{Theorem}
\newtheorem{prop}{Proposition}
\newtheorem{procedure}{Procedure}

\usepackage{array}
\usepackage{authblk}
\usepackage{setspace}
\usepackage{fullpage,etoolbox}

\usepackage{epsfig} 
\usepackage{optidef}

\usepackage[font={small}]{caption}
\usepackage{graphicx}
\usepackage{xcolor}
\usepackage{comment}
\usepackage{url}
\newcounter{num}

\newcommand{\rnum}[1]{\setcounter{num}{#1} \roman{num}}
\usepackage{amsmath,amssymb,amsfonts}
\usepackage{mathtools}
\usepackage{subcaption}
\usepackage{algorithm}
\usepackage[noend]{algpseudocode}
\usepackage{color}
\usepackage{multirow}

\algrenewcommand\algorithmicrequire{\textbf{Input:}}
\algrenewcommand\algorithmicensure{\textbf{Output:}}

\algnewcommand{\Break}{\textbf{break}}

\algnewcommand\algorithmicforeach{\textbf{for each}}
\algdef{S}[FOR]{ForEach}[1]{\algorithmicforeach\ #1\ \algorithmicdo}

\makeatletter
\newcommand*{\itemequation}[3][]{%
  \item
  \begingroup
    \refstepcounter{equation}%
    \ifx\\#1\\%
    \else  
      \label{#1}%
    \fi
    \sbox0{#2}%
    \sbox2{$\displaystyle#3\m@th$}%
    \sbox4{\@eqnnum}%
    \dimen@=.5\dimexpr\linewidth-\wd2\relax
    \ifcase
        \ifdim\wd0>\dimen@
          \z@
        \else
          \ifdim\wd4>\dimen@
            \z@
          \else 
            \@ne
          \fi 
        \fi
      \@latex@warning{Equation is too large}%
    \fi
    \noindent   
    \rlap{\copy0}%
    \rlap{\hbox to \linewidth{\hfill\copy2\hfill}}%
    \hbox to \linewidth{\hfill\copy4}%
    \hspace{0pt}
  \endgroup
  \ignorespaces 
}
\makeatother

\pgfplotsset{compat=1.10}
 \def\BibTeX{{\rm B\kern-.05em{\sc i\kern-.025em b}\kern-.08em
    T\kern-.1667em\lower.7ex\hbox{E}\kern-.125emX}}
\usepackage{capt-of}
\usepackage{pgfplots}
\usepackage{fancyhdr}
\usepackage{xcolor}
\usepackage{wasysym,xcolor}
\usepackage{hyperref}

\def\doi{10.1109/TAC.2025.3555949}
\fancyfootoffset{+0.0em}

\fancypagestyle{copyright}{
    \vspace{-2.0em}
	\cfoot{\vspace{0.0em}\footnotesize This work was partially supported by JSPS KAKENHI under Grants JP22H01510, 21K14184, 22KK0155, and JP23K22780, as well as by JST under Grant CREST JPMJCR201 and ACT-X JPMJAX23CK. \\ Corresponding author: Y. Takayama. Email:(yoshinari.takayama@centralesupelec.fr, \\hashimoto@eei.eng.osaka-u.ac.jp, ohtsuka@i.kyoto-u.ac.jp) }
	\chead{\footnotesize This article has been accepted for publication in IEEE Transactions on Automatic Control. This is the author's one-column version, which has not been fully edited, and the content may change before final publication. Citation information: \href{http://dx.doi.org/10.1109/TAC.2025.3555949}{DOI \doi}}
}

\title{\bf STLCCP: Efficient Convex Optimization-Based Framework for Signal Temporal Logic Specifications}
\author{Yoshinari~Takayama$^\dagger$, Kazumune Hashimoto$^\diamond$, Toshiyuki Ohtsuka$^\ast$
}

\affil{$^\dagger$ \textit{Laboratory of Signals and Systems (L2S), CentraleSupélec, University of Paris-Saclay}}

\affil{$^\diamond$ \textit{Graduate School of Engineering, Osaka University}}

\affil{$^\ast$ \textit{Graduate School of Informatics, Kyoto University}}

\date{}

\begin{document}

\maketitle
\thispagestyle{copyright}
\pagestyle{empty}
\begin{abstract}
Signal temporal logic (STL) is a powerful formalism for specifying various temporal properties in dynamical systems. However, existing methods, such as mixed-integer programming and nonlinear programming, often struggle to efficiently solve control problems with complex, long-horizon STL specifications. This study introduces \textit{STLCCP}, a novel convex optimization-based framework that leverages key structural properties of STL: monotonicity of the robustness function, its hierarchical tree structure, and correspondence between convexity/concavity in optimizations and conjunctiveness/disjunctiveness in specifications. The framework begins with a structure-aware decomposition of STL formulas, transforming the problem into an equivalent difference of convex (DC) programs. This is then solved sequentially as a convex quadratic program using an improved version of the convex-concave procedure (CCP). To further enhance efficiency, we develop a smooth approximation of the robustness function using a function termed the \textit{mellowmin} function, specifically tailored to the proposed framework. Numerical experiments on motion planning benchmarks demonstrate that \textit{STLCCP} can efficiently handle complex scenarios over long horizons, outperforming existing methods.
\end{abstract}

\section{Introduction}
Autonomous robotic systems, such as self-driving cars, drones, and cyber-physical systems, can have extensive real-world applications in transportation, delivery, and industrial automation in the near future. These systems must perform complex control tasks while ensuring safety, ideally with formal guarantees. To address this issue, researchers have developed formal specification languages that offer a unified and rigorous approach to expressing complex requirements for safe and reliable autonomous systems.
These languages include signal temporal logic (STL)~\cite{oded_stl}, which is well-suited for continuous real-valued signals. It offers a useful quantitative semantics called \textit{robustness function}~\cite{Fainekos2009-yn}, which quantifies how robustly a formula is satisfied. By maximizing this score, control input sequences can be synthesized robustly, and the resulting trajectories can formally satisfy the specification.
However, precisely formulating this optimization problem as a mixed-integer program (MIP) faces scalability issues with respect to the length of the horizon. To avoid this issue, recent work has focused on formulating the problem as nonlinear programs (NLP) by smooth approximations of the~$\operatorname{\max}$ and~$\operatorname{\min}$ operators in the robustness function~\cite{Pant_smooth, Hashimoto2022-xu,Gilpin2021-wv}. These NLP are then solved \textit{naively} through a sequential quadratic programming (SQP) method (or other gradient-based methods). Although these SQP-based methods are generally faster and more scalable than MIP-based methods, they can easily become infeasible and may not lead to the global optima due to the coarse iterative approximations. For instance, the traditional SQP method approximates the original problem as a quadratic program, which captures only up to the second derivative information of the original problem. This drawback becomes even more pronounced in STL-related problems, as the robustness function is coarsely approximated due to its high degree of non-convexity.

This study addresses these challenges by explicitly integrating key \textit{properties of STL} into the optimization framework. These features include the monotonicity property of the robustness function, its hierarchical tree structure, and correspondence between convexity/concavity in optimization and conjunctive/disjunctive operators in specifications. Leveraging these features, we propose a novel optimization framework, \textit{STLCCP}, which employs the convex-concave procedure (CCP)~\cite{Lipp2016-fa,Shen2016-sn} for iterative optimization. 

This study makes three primary contributions.
\begin{itemize}
\item[\rnum{1})] \textbf{Robustness Decomposition Framework:} We introduce an efficient decomposition method that reformulates the original problem into a specific difference of convex (DC) program. This approach reduces the non-convex components of the problem by leveraging the key properties of STL, enabling the effective use of CCP during the subsequent optimization step. 

\item[\rnum{2})] \textbf{Tree-Weighted Penalty CCP (TWP-CCP):} We develop an enhanced CCP variant that prioritizes critical constraints by exploiting the hierarchical structure of STL, which improves the optimization step.

\item[\rnum{3})] \textbf{Mellowmin-based Robustness Approximation:} To enhance optimization with gradient-based solvers, we propose using a novel smoothing function, the \textit{mellowmin} function. This robustness measure is both sound and asymptotically complete while providing a trade-off between smoothing accuracy and reduced concavity, making it appropriate for the proposed framework. Furthermore, we introduce a warm-start strategy to efficiently utilize this measure in practice.
\end{itemize}
These contributions form the foundation of the STLCCP framework, providing a scalable and robust solution for synthesizing control strategies in autonomous systems.

The STLCCP framework can handle \textit{any} temporal operators and nesting depth for STL specifications while \textit{ensuring} the formula's satisfaction, even when incorporating the constraint relaxation from TWP-CCP and the smooth approximation introduced via the mellowmin function. Moreover, it is computationally efficient: In contrast to the conventional SQP method, the proposed framework retains complete information of convex parts at each iteration, resulting in fewer approximations. In particular, when a linearity assumption is satisfied, it solves quadratic programs sequentially, with only disjunctive parts of the specification being linearized at each iteration. Numerical experiments show that the proposed method outperforms existing methods in terms of both robustness score and computational time on several benchmarks.  

\subsection{Related works}\label{subsection:relatedworks}
\textbf{Control under temporal logic specifications:} Traditional approaches to controller synthesis under temporal logic specifications have predominantly relied on model abstraction techniques. Although these methods are effective, they demand substantial domain expertise, posing challenges to automation and scalability. Optimization-based approaches have emerged as a compelling alternative. Early works~\cite{karaman2008,Karaman2008OptimalCO} concentrated on linear temporal logic (LTL) and metric temporal logic (MTL), while recent studies have extended these methodologies to STL, enabling their application to robotic systems with continuous state spaces~\cite{Raman2014-sa}. Most of these studies have focused on either developing novel robustness metrics~\cite{Pant_smooth,Gilpin2021-wv,Mehdipour2020-zm} or proposing an efficient technique for a specific class of temporal logic specifications~\cite{Lindemann2019-os}. More recently, research has sought to exploit the hierarchical structure of STL to improve MIP-based algorithms. For instance,~\cite{Kurtz2022-pe} introduces techniques to minimize integer variables, while ~\cite{Sun2022-fg} proposes novel encoding methods for multi-agent problems. Nevertheless, these methods target MIP formulations and do not address NLP-based approaches. Moreover, they do not fully exploit the key properties of STL, as this study aims to do (we refer to Section~\ref{subsec:discuss} for more detailed discussions).

\textbf{Control using convex optimizations:}
Convex optimization frameworks, such as CCP, have been applied in optimal control in several works, including~\cite{Debrouwere2013-ds_diehl,2011diehl}. More recently, these frameworks have been utilized in interdisciplinary research bridging formal methods and control theory. For instance,~\cite{Cubuktepe2021ConvexOF} addressed parameter synthesis problems in MDPs, while~\cite{WANG2022104965} derived invariant barrier certificate conditions to ensure unbounded time safety. Despite the growing popularity of convex optimization, its application to STL specifications remains relatively unexplored. Although several convex optimization-based frameworks for STL specifications exist, such as~\cite{maria,Mao2022-rq}, their focuses differ from our approach. For instance,~\cite{maria} emphasizes the decomposition of global STL specifications for multi-agent systems into local STL tasks, whereas our framework addresses controller synthesis for a single system. Similarly,~\cite{Mao2022-rq} employs a convex optimization-based approach for STL specifications, but their solver constrains updates within a trust region, differing from CCP-based approaches. Moreover, the approximated components of their method are fundamentally distinct from our approach.


In the conference paper~\cite{stlccp}, we explored a method for applying CCP to control problems with STL specifications. However, this previous work was limited to affine systems and predicate functions. In contrast, this paper generalizes the approach to encompass general DC function cases, including non-smooth and nonaffine functions. Additionally, it provides a comprehensive characterization of the systems and predicate functions to which the method can be applied. This study also introduces the TWP-CCP approach and the mellowmin function, both of which were absent in the earlier conference paper. Integrating TWP-CCP with the mellowmin function enables efficient use of gradient-based solvers while ensuring formula satisfaction, halving computation time and enabling consistent achievement of high robustness scores, regardless of the initial guesses.

The rest of the paper is structured as follows. Section~\ref{sec:preparation} introduces the preliminaries and the problem statement, presenting two key concepts that bridge convex optimization and formal specifications: the reversed robustness function and its tree structure. Section~\ref{section:decomposition} describes the robustness decomposition method, reformulating the optimization problem into a non-smoothed DC form. Section~\ref{section:CCP} focuses on the optimization step and examines the properties of standard CCP and the proposed TWP-CCP. Section~\ref{section:smooth} explores the differentiable case, proposing the mellowmin robustness measure. Section~\ref{section:example} presents numerical experiments to validate the approach. Finally, Section~\ref{section:conclusion} discusses how the framework leverages the key properties of STL and concludes the paper.

\section{Preliminaries}\label{sec:preparation}
\paragraph*{Notations}~$\mathbb{R}$, $\mathbb{R}_{>0}$, and $\mathbb{N}$ represent the sets of real numbers, positive real numbers, and natural numbers, respectively. 
Other similar sets are defined analogously. 
A signal, or trajectory, $\boldsymbol{x}:\mathbb{T}\to\mathbb{R}^{n}$, is a function that maps each time point $t$ in the time domain $\mathbb{T} \subseteq \mathbb{N}$ to an $n$-dimensional real-valued vector $x_t \in \mathbb{R}^n$. For simplicity, we consider a discretized time domain with $H\in\mathbb{N}$ time steps (horizons), denoted as $\mathbb{T} = \{0, \ldots, H\}$, and represent a trajectory over $\mathbb{T}$ as the sequence~$\boldsymbol{x} = \{x_0, x_1, \ldots, x_H\}$. A time interval~$[t_1,t_2] \subseteq \mathbb{T}$ denotes a segment of the time domain, with~$0 \leq t_1 \leq t_2$. 

\subsection{Signal temporal logic (STL)}\label{subsec:stl}

STL is a predicate logic used to specify properties of trajectories~$\boldsymbol{x}$. The syntax of STL is~\cite{Raman2014-sa}:
\begin{align}\label{eq;stl}
\varphi:=\mu| \vee \varphi | \wedge \varphi|\square_{\left[t_1, t_2\right]} \varphi| \Diamond_{\left[t_1, t_2\right]} \varphi \mid \varphi_1 \boldsymbol{U}_{\left[t_1, t_2\right]} \varphi_2,
\end{align}
where~a \textit{predicate} is expressed as~$\mu=(g^\mu(x_t) \leq 0)$, with~$g^\mu:\mathbb{R}^{n}\rightarrow\mathbb{R}$ referred to as the \textit{predicate function}. The symbol~$|$ stands for OR and the definition is recursive.
In addition to the standard boolean operators~$\wedge$ and~$\vee$, the STL also incorporates temporal operators~$\square$~(\textit{always}),~$\Diamond$~(\textit{eventually}), and~$\boldsymbol{U}$~(\textit{until}). Here,~$\wedge$ and~$\square$ are referred to as conjunctive operators. Similarly,~$\vee$ and~$\Diamond$ are referred to as disjunctive operators. 
Any STL formula can be brought into negation normal form~(NNF), where the negations are solely present in front of the predicates~\cite{sadradd_monotone,Sadraddini2015-dn}. This study deliberately omits negations from the STL syntax for later explanation. For simplicity, we refer to this negation-free STL as STL.

Each STL formula is associated with its quantitative semantics, known as the \textit{robustness function}. This function indicates how well the STL specification is satisfied or violated. For the definition of the traditional robustness function at time~$t$ for a given trajectory~$\boldsymbol{x}$ and formula~$\varphi$, denoted as~$\rho^\varphi(\boldsymbol{x}, t)\in\mathbb{R}$, we refer readers to~\cite[Sec.~2-C]{Raman2014-sa}. Here, we adopt the same robustness measure but redefine it to align with our convex optimization framework. Specifically, we introduce a reversed version of the robustness function, denoted as~\( \rho^\varphi_{\text{rev}}(\boldsymbol{x}, t)\in\mathbb{R} \), which is recursively defined by the following equations:

\begin{defi}\textit{(STL Reversed robustness)}\label{def:revrobustness}
\begin{subequations}\label{eq:robustnessrev}
\begin{align}
    & \rho_{\text{rev}}^\mu(\boldsymbol{x}, t)=g^\mu\left(x_t\right) \label{eq:reversed_predicate}\\
    & \rho_{\text{rev}}^{\varphi_1 \wedge \varphi_2}(\boldsymbol{x}, t)=\max \left(\rho_{\text{rev}}^{\varphi_1}(\boldsymbol{x}, t), \rho_{\text{rev}}^{\varphi_2}(\boldsymbol{x}, t)\right) \label{eq:reversed_conj}\\
    & \rho_{\text{rev}}^{\varphi_1 \vee \varphi_2}(\boldsymbol{x}, t)=\min \left(\rho_{\text{rev}}^{\varphi_1}(\boldsymbol{x}, t), \rho_{\text{rev}}^{\varphi_2}(\boldsymbol{x}, t)\right)\\
    & \rho_{\text{rev}}^{\square_{\left[t_1, t_2\right] }\varphi}(\boldsymbol{x}, t)=\max _{t^{\prime} \in\left[t+t_1, t+t_2\right]}\left(\rho_{\text{rev}}^{\varphi}\left(\boldsymbol{x}, t^{\prime}\right)\right)\\
    & \rho_{\text{rev}}^{\Diamond_{\left[t_1, t_2\right] }\varphi}(\boldsymbol{x}, t)=\min _{t^{\prime} \in\left[t+t_1, t+t_2\right]}\left(\rho_{\text{rev}}^{\varphi}\left(\boldsymbol{x}, t^{\prime}\right)\right) \label{eq:reversed_eventually}\\
    & \hspace{-0.2cm} \rho_{\text{rev}}^{\varphi_1 \boldsymbol{U}_{\left[t_1, t_2\right]} \varphi_2}(\boldsymbol{x}, t)=\max _{t^{\prime} \in\left[t+t_1, t+t_2\right]}\Bigl(\min \Bigl(\rho_{\text{rev}}^{\varphi_1}(\boldsymbol{x}, t^{\prime}), \nonumber \\
    &  \min _{t^{\prime \prime} \in[t+t_1, t^{\prime}]}\bigl(\rho_{\text{rev}}^{\varphi_2}(\boldsymbol{x}, t^{\prime \prime})\bigr)\Bigr)\Bigr) \label{eq:reversed_until}
\end{align}
\end{subequations}
\end{defi}
In Definition~\ref{def:revrobustness}, each~$\max$ and~$\min$ function in~\eqref{eq:reversed_conj}--\eqref{eq:reversed_until} is swapped compared to those in the traditional robustness definition in~\cite{Raman2014-sa}. Using this reversed robustness, the traditional robustness can then be obtained as~$\rho^{\varphi}(\boldsymbol{x},t) = -\rho_{\text{rev}}^{\varphi}(\boldsymbol{x},t)$. For simplicity, we abbreviate $\rho^{\varphi}(\boldsymbol{x},0)$ as $\rho^{\varphi}(\boldsymbol{x})$.
The horizon needed to calculate the robustness of a formula, known as \textit{formula length}, can be calculated recursively as in~\cite{Sadraddini2015-dn}. 
The robustness~$\rho^{\varphi}(\boldsymbol{x}) (= -\rho_{\text{rev}}^{\varphi}(\boldsymbol{x}))$ is \textit{sound} in the sense that 
\begin{equation}\label{eq:sound}
    \rho_{\text{rev}}^{\varphi}(\boldsymbol{x}) \leq 0 \implies   \boldsymbol{x} \vDash \varphi,
\end{equation} 
and also \textit{complete} in the sense that 
$\rho_{\text{rev}}^{\varphi}(\boldsymbol{x}) > 0 \implies   \boldsymbol{x} \nvDash \varphi$.\footnote{Throughout this paper, we use \textit{soundness} to describe when a robustness value can guarantee the satisfaction of a formula, as in~\eqref{eq:sound}. On the other hand, we use \textit{completeness} to describe when a robustness value can also guarantee the violation of a formula. These terminologies are the same as~\cite{Gilpin2021-wv} but are different from~\cite{Mehdipour2020-zm}.}
\begin{remark}\label{rem:redundunt}
Although this redefinition may seem unnecessary, we introduce it for several reasons. First, it establishes a clear correspondence between temporal logic operator types (conjunctive or disjunctive) and optimization constraint curvatures (convex or concave). 
Table~\ref{tab:connection} summarizes this correspondence. It shows that in the reversed version, satisfying conjunctive operators leads to satisfying a convex constraint. 
For instance, consider a formula~$\varphi = \varphi_1 \wedge \varphi_2$, where~$\varphi_1$ and~$\varphi_2$ are linear predicates. When expressed in terms of the reversed robustness~$\rho_{\text{rev}}^{\varphi}$, this correspondence implies that the conjunction~$\wedge$ corresponds to~$\max$, and the constraint becomes~$\max(g^{\varphi_1}, g^{\varphi_2}) \leq 0$. In contrast, using the traditional robustness, the constraint is~$-\min(-g^{\varphi_1}, -g^{\varphi_2}) \leq 0$, where such a direct connection is not immediately apparent. Second, this definition avoids maximization formulations, which are not commonly used in convex optimization, reducing confusion. This helps in intuitively understanding the connection between temporal logic and convex optimization. Finally, employing this reversed definition simplifies the decomposition procedure, which will be discussed in Section~\ref{section:decomposition}.
\begin{table}[ht]
\begin{center}
\normalsize
\caption{Correspondance in the reversed robustness.}\label{tab:connection}
\begin{tabular}{|c||c||c||c|}
\hline
 node-type  & operator & functions & constraints \\ \hline \hline
 conjunctive &~$\wedge, \square$ & max  & convex \\ \hline
 disjunctive &$\vee,\Diamond$ & min & concave \\ \hline
\end{tabular}
\end{center}
\end{table}
\normalsize 
\end{remark}

\subsection{Tree structure}\label{subsec:tree_structure}
STL formulas have a tree structure. We define \textit{robustness tree}~$\mathcal{T}^{\varphi}$ for the reversed robustness functions~$\rho_{\text{rev}}^{\varphi}$:

\begin{defi}\textit{(Robustness tree)}\label{def:tree}
A robustness tree~$\mathcal{T}^{\varphi}$ is a tuple~$\left(\mathcal{A},\mathcal{B}, \mathcal{C}\right)$ corresponding to~$\rho_{\text{rev}}^{\varphi}$, where
\begin{enumerate}
\item[\rnum{1})]~$\mathcal{A}\in\{\max, \min\}$ is the type of the top node of~$\mathcal{T}^{\varphi}$, which corresponds to the outermost operator of~$\rho_{\text{rev}}^{\varphi}$.
\item[\rnum{2})]~$\mathcal{B}=\left[\mathcal{T}^{\varphi_1}, \mathcal{T}^{\varphi_2}, \ldots, \mathcal{T}^{\varphi_T}\right]$ is a list of~$T$ first-level subtrees, called \textit{child subtrees}, sprouted from the top node, which corresponds to each argument of~$\rho_{\text{rev}}^{\varphi}$.
\item[\rnum{3})]~$ \mathcal{C}=\left[t^{\varphi_1}, t^{\varphi_2}, \ldots, t^{\varphi_T}\right]$ is a list where each element~$t^{\varphi_i}$ represents the time step associated with the first-level subtree $\mathcal{T}^{\varphi_i}\in \mathcal{B}$. For when the top node of $\mathcal{T}^{\varphi}$ is associated with a temporal operator, $t^{\varphi_i}$ corresponds to each time step within its time interval, while for Boolean operators, all $t^{\varphi_i}$ are $0$.
\end{enumerate}
\end{defi}
The leaves of a robustness tree $\mathcal{T}^{\varphi}$, which are the bottom nodes, correspond to the predicate functions in~$\rho_{\text{rev}}^{\varphi}$. Each subtree of $\mathcal{T}^{\varphi}$ is itself a robustness tree, and every node in $\mathcal{T}^{\varphi}$, except the leaves, serves as the top node of a subtree. Given $\rho_{\text{rev}}^\varphi$ (or the formula $\varphi$ itself), the tree $\mathcal{T}^{\varphi}$ and all of its subtrees can be constructed recursively. An example of a robustness tree is shown in Fig.~\ref{fig:robustnesstree}.

\begin{figure}[!ht]
    \centering
    \includegraphics[keepaspectratio, scale=0.5]{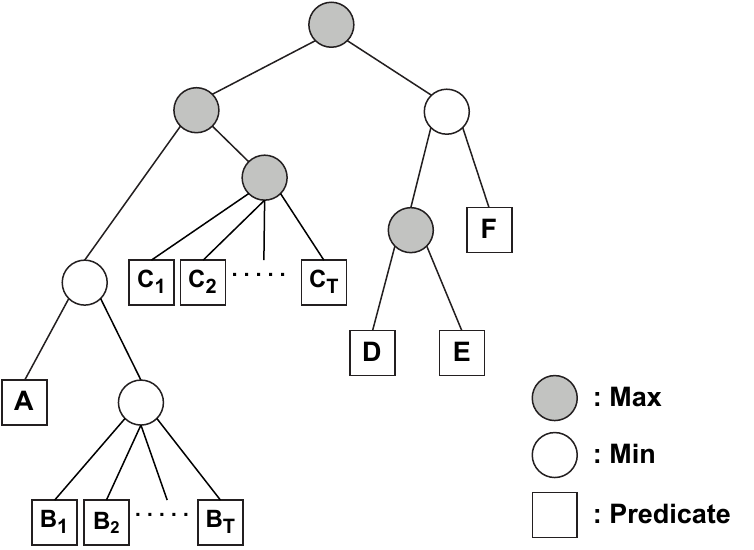}
    \caption{An example of the robustness tree for formula \protect{$\varphi=((A\vee\Diamond_{[1, T]}B)\wedge \square_{[1, T]}C) \wedge((D\wedge E)\vee F)$}, where~$A,B_1,...,B_T,C_1,...,C_T,D,E,F$ denote predicates. The grey circles represent the \protect{$\max$} operators, whereas the white circles represent the \protect{$\min$} operators. The square boxes represent predicate functions~$g^A,...,g^F:\mathbb{R}^{n}\rightarrow\mathbb{R}$ associated with each predicate.}
    \label{fig:robustnesstree}
\end{figure}
Each~$\max$-type node (or~$\min$-type node) in~$\mathcal{T}^{\varphi}$ corresponds to an operator~$\wedge$ or~$\square$ (or~$\vee$ and~$\Diamond$) in the formula~$\varphi$, as illustrated in Table~\ref{tab:connection}. Here, the operators~$\wedge$ and~$\square$ (or~$\vee$ and~$\Diamond$) are treated similarly. However, the ``until" operator~$\boldsymbol{U}$ does not correspond to a single node as its robustness function in~\eqref{eq:reversed_until} is a conjunction of~$\max$ and~$\min$ functions. Specifically, it corresponds to a robustness tree whose top node is of~$\max$-type, with child nodes that are~$\min$-type, and those in turn may be further subdivided into~$\min$-type nodes or predicates.

Here, we denote the descendant trees of~$\mathcal{T}^\varphi$ whose top node is of~$\min$-type as~$\mathcal{T}^\varphi_{\vee,i}$, with~$i \in \{1, \ldots, N_\vee^\varphi\}$ and~$N_\vee^\varphi\in\mathbb{N}$ representing the total number of such descendant trees. Since the top node of each~$\mathcal{T}^\varphi_{\vee,i}$ corresponds to a disjunctive operator in the STL formula~$\varphi$, we refer to them as \textit{disjunctive nodes} of~$\mathcal{T}^\varphi$. Each leaf (predicate) in~$\mathcal{T}^\varphi$ is denoted as~$\mathcal{L}^\varphi_i$, where~$i \in \{1, \ldots, N_p^\varphi\}$ and~$N_p^\varphi\in\mathbb{N}$ denotes the total number of predicates in~$\varphi$. The set of all predicate functions in~$\varphi$ is denoted as~$\mathcal{M} = \{\mu_1, \ldots, \mu_{N_p^\varphi} \}$. Furthermore, we classify the leaves based on their parent nodes. Leaves with~$\max$-type parent nodes are denoted as~$\mathcal{L}^\varphi_{\wedge,i}$, with~$i \in \{1, \ldots, N_{p \wedge}^\varphi\}$, where~$N_{p \wedge}^\varphi\in\mathbb{N}$ represents the number of such leaves. Similarly, leaves with~$\min$-type parent nodes are denoted as~$\mathcal{L}^\varphi_{\vee,i}$, with~$i \in \{1, \ldots, N_{p \vee}^\varphi\}$, where~$N_{p \vee}^\varphi\in\mathbb{N}$ is the count of such leaves. Finally, it holds that the total number of predicates satisfies the relationship~$N_p^\varphi = N_{p \wedge}^\varphi + N_{p \vee}^\varphi$.

\begin{example}
Consider the formula shown in Fig.\ref{fig:robustnesstree}. For the robustness tree~$\mathcal{T}^\varphi = (\mathcal{A} = \max, \mathcal{B}=[\mathcal{T}^{(A\vee\Diamond_{[1, T]}B)\wedge \square_{[1, T]}C},\mathcal{T}^{(D\wedge E)\vee F}], \mathcal{C}=\left[0,0\right])$, one of its disjunctive subtrees,~$\mathcal{T}^{\Diamond_{[1, T]} B}$, has~$T$ child subtrees with a time step list~$\left[0, 1, \ldots, T-1\right]$. In total, there are~$N^{\varphi}_{\vee}=3$ disjunctive subtrees. Among all~$N^{\varphi}_{p}=2T+4$ predicates, those with~$\min$-type parent nodes number~$N^{\varphi}_{p \vee}=T+2$.
\end{example}
\begin{remark}
The authors of~\cite{Kurtz2022-pe,Sun2022-fg,Leung2019-back} also define a tree for STL formulas. However, they define the tree for the formula~$\varphi$ itself, whereas we define it for the reversed robustness function~$\rho_{\text{rev}}^{\varphi}$ rather than directly for~$\varphi$.
\end{remark}

\subsection{Convex-concave procedure (CCP)}\label{subsec:ccp}
The CCP~\cite{Shen2016-sn} is an iterative algorithm to find a local optimum of a class of optimization problems called DC programs, which are composed of functions called DC functions:
\begin{defi}\textit{\cite{Lipp2016-fa}} \label{def:dcf}
A function~$p:\mathbb{R}^h\rightarrow\mathbb{R}$ is a DC function if there exist convex functions~$q, r: \mathbb{R}^h \rightarrow \mathbb{R}$ such that~$p$ is expressed as the difference between~$q$ and~$r$, that is,~$p(z)=q(z)-r(z),z\in\mathbb{R}^h$.
\end{defi}
\begin{defi}\textit{\cite{Lipp2016-fa}}
A program is a DC program if written as
\begin{subequations}\label{eq:dc}
\begin{align}
&\min _{z} \quad p_0(z) \label{eq:dc1}\\
&\text{s.t. } \quad p_i(z)\leq 0, \quad i\in\{1, \ldots, m\},\label{eq:dc2}
\end{align}
\end{subequations}
where~$z\in\mathbb{R}^h$ is the vector of~$h$ optimization variables and~$p_i=q_i-r_i: \mathbb{R}^h \rightarrow \mathbb{R},i\in\{0, \ldots, m\}$ are DC functions. Note that constraints in \eqref{eq:dc2} can include equality constraints
\begin{equation}\label{eq:dc3}
q_i(z)=r_i(z)
\end{equation}
by encoding them as a pair of inequality constraints
$q_i(z)-~r_i(z)~\leq~0$ and~$r_i(z)-q_i(z) \leq 0.$
\end{defi}
The class of DC programs is broad; for example, any~$C^2$ function can be expressed as a DC function~\cite{Hartman1959p707}. Moreover, it also includes some non-differentiable functions. 

The CCP involves two steps: a majorization step, where concave terms are replaced with a convex upper bound, and a minimization step, where the resulting convex problem is solved. This study adopts the simplest form of majorization, which is linearization at the current point at each iteration. Specifically, the CCP linearizes all \textit{(nonaffine) concave functions} in the following forms: \rnum{1}) function~$r_0(z)$ in the cost~\eqref{eq:dc1}; \rnum{2}) functions~$r_i(z)$ in inequality constraints~\eqref{eq:dc2};
\rnum{3}) functions~$q_i(z)$ \textit{ and }~$r_i(z)$ in equality constraints~\eqref{eq:dc3}.

Although SQP algorithms have been popular for their per-iteration solving as quadratic programs, advances in algorithms for tackling more general convex programs, such as second-order cone programs (SOCPs), have made it possible to exploit additional problem structure and potentially achieve better solutions through CCP. Although the computation time per iteration in CCP can be longer due to solving general convex programs, CCP can outperform SQP in terms of total computation time when the number of iterations required for convergence is sufficiently small.

\subsection{Problem formulation}
We consider a nonlinear system
\begin{equation}\label{eq:system}
x_{t+1} = f(x_t,u_t), 
\end{equation}
where~$x_t \in \mathcal{X} \subseteq \mathbb{R}^n$ and~$u_t \in \mathcal{U} \subseteq \mathbb{R}^m$ are the state and control input at time~$t$, and~$f: \mathcal{X} \times \mathcal{U}\rightarrow \mathcal{X}$. The state and input constraints~$\mathcal{X}$ and $\mathcal{U}$ are defined by the inequalities~$h_x(x_t)\leq 0$ and~$h_u(u_t)\leq 0$, respectively, where $h_x : \mathbb{R}^n \to \mathbb{R}^{l_{x_t}}$, $h_u : \mathbb{R}^m \to \mathbb{R}^{l_{u_t}}$, and $l_{x_t}, l_{u_t}$ correspond to the number of state and input constraints, respectively, for $t\in\{1,...,H\}$.

The control problem to be addressed, denoted as~$\mathcal{P}$, is defined in Problem~\ref{problem}.
\begin{problem}\label{problem}
Given an initial state~$x_0$, systems~\eqref{eq:system}, and STL specification~$\varphi$, find a trajecotry~$\boldsymbol{x}$ that robustly satisfies~$\varphi$:
\begin{subequations}\label{eq:moto}
\begin{align}
\min _{\boldsymbol{x}, \boldsymbol{u}} \hspace{0.5em}&\rho_{\text{rev}}^{\varphi}(\boldsymbol{x}) \label{eq:min_x1}\\
\text { s.t. } & x_{t+1}=f(x_t,u_t), \quad t\in\{1,...,H\}\label{eq:min_sys1}\\
& x_t \in \mathcal{X}, u_t \in \mathcal{U}, \quad t\in\{1,...,H\} \label{eq:min_xu1}\\
& \rho_{\text{rev}}^{\varphi}(\boldsymbol{x})\leq 0
\label{eq:min_xuu1}
\end{align}
\end{subequations}
\end{problem}

The proposed framework adopts the following two assumptions regarding the system, state/input constraints, and STL specifications throughout this paper. In Section~\ref{subsec:result}, however, we introduce additional assumptions—Assumptions~\ref{assum:major} and~\ref{assum:linear}—which provide useful properties of the proposed framework. While the predicate functions are restricted, there are no limitations on the types of temporal operators or the nesting depth, allowing the general syntax in~\eqref{eq;stl}.
\begin{assumption}\label{assum:dc}
\textit{(System and state/input constraints)} The functions~$f,h_x,h_u$ are all DC functions (with respect to their arguments: both~\(x_t, u_t\) for~\(f\),~\(x_t\) for~\(h_x\), and~\(u_t\) for~\(h_u\)).
\end{assumption}
\begin{assumption} \label{assum:dcpredicate}
\textit{(STL predicate functions)} If the parent node of predicate~$\mu_i \in \mathcal{M}$ is of~$\min$-type, the predicate function~$g^{\mu_i}$ is a concave function (with respect to its arguments $x_t$). However, if the parent node of predicate~$\mu_i$ is of~$\max$-type, the predicate function~$g^{\mu_i}$ is a general DC function.
\end{assumption}

Assumption~\ref{assum:dcpredicate} is a specific requirement of the proposed approach, necessary to reformulate the program as a DC program.
We provide an example of this assumption below.
\begin{example}
Furthermore, consider the formula in Fig.~\ref{fig:robustnesstree}.
The predicate functions of the~$N^{\varphi}_{p \vee}$ predicates, that is~$g^A,g^{B_1},...,g^{B_T},g^F$, can be concave functions. However, the other~$N^{\varphi}_{p \wedge}$ predicate functions, that is~$g^{C_1},...,g^{C_T},g^D,g^E$, can be general DC functions.  
\end{example}
\begin{remark}\label{rem:assumpredicate}
Assumption~\ref{assum:dcpredicate} depends only on the parent node and does not extend to the grandparent node or any other ancestral nodes. This assumption accommodates a broader class of predicate functions than affine functions, including non-differentiable ones. Although this study mainly focuses on differentiable functions for simplicity of notation, the results remain valid for non-differentiable cases, where the gradient at a point is replaced by a subgradient at that point. 
\end{remark}
\begin{remark}\label{rem:omitqudratic}
Although we omit quadratic cost terms (such as those in the LQR setting), output equations, and other variables for simplicity in \eqref{eq:moto}, the proposed method can still be applied to such cases if the resulting program is a DC program. We could also omit Eq.~\eqref{eq:min_xuu1}. 
In such cases, the program maximizes robustness without the guarantee of the formula's satisfaction.
\end{remark}

\section{Robustness Decomposition}\label{section:decomposition}

The decomposition aims to turn the problem into a DC program to apply a CCP-based algorithm.
Achieving this requires determining the convexity or concavity of each function in the program~\eqref{eq:moto}. Among these, the reversed robustness function~$\rho_{\text{rev}}^{\varphi}$—a composite of~$\max$ and~$\min$ functions—is the only component with an indeterminate curvature (convex or concave). To address this issue, we decompose the robustness function into convex and concave functions \textit{equivalently}.\footnote{In this paper, two optimization problems are called \textit{equivalent} if their (global) optimal objective values and solutions can be obtained from those of the other. This is a similar usage as in \protect{\cite[p.257]{El_Ghaoui2015-eb}} and \protect{\cite[p.130]{Boyd2004-kn}}.}

\subsection{Basic idea of the proposed robustness decomposition}\label{subsection:first_step_of_}
We begin by demonstrating that program~\eqref{eq:moto} can be equivalently transformed into a DC program using a straightforward robustness decomposition procedure, under a mild additional assumption on predicate functions:
\begin{assumption} \label{assum:maxmin}
\textit{(STL predicate functions)} If the parent node of predicate~$\mu_i\in \mathcal{M}$ is of~$\max$-type, the predicate function~$g^{\mu_i}$ is a convex function. 
\end{assumption}

\begin{prop}\label{theo:dcprogram}
Program~\eqref{eq:moto} can be equivalently transformed into a DC program under Assumptions~\ref{assum:dc}--\ref{assum:maxmin}.
\end{prop}
\begin{proof}
We reformulate the program by equivalently decomposing the robustness function. The outermost operator of the robustness function~$\rho_{\text{rev}}^{\varphi}$ is either~$\max$ and~$\min$. 
In both cases, we replace the cost function with variable~$s_\xi \in\mathbb{R}$. For instance, if the outermost operator is the~$\max$, that is,~$\rho_{\text{rev}}^{\varphi} = \max(\rho_{\text{rev}}^{\varphi_1},\rho_{\text{rev}}^{\varphi_2},...,\rho_{\text{rev}}^{\varphi_r})$, we put~$s_\xi$ in the cost and add the constraint
\begin{equation}\label{eq:basic:max}
\max(\rho_{\text{rev}}^{\varphi_1},\rho_{\text{rev}}^{\varphi_2},...,\rho_{\text{rev}}^{\varphi_r})=s_\xi.
\end{equation}
In the~$\min$ case, we add 
\begin{equation}\label{eq:basic:min}
\min(\rho_{\text{rev}}^{\varphi_1},\rho_{\text{rev}}^{\varphi_2},...,\rho_{\text{rev}}^{\varphi_r})=s_\xi,
\end{equation}
where the formulas $\varphi_1, \ldots, \varphi_r$ in the constraints~\eqref{eq:basic:max} and~\eqref{eq:basic:min} are not necessarily predicates.
Subsequently, for both constraints~\eqref{eq:basic:max} and \eqref{eq:basic:min}, we verify whether each argument is a predicate function or not. If an argument of~$\max$ function in~\eqref{eq:basic:max} is not a predicate function, say~$\rho_{\text{rev}}^{\varphi_i}$, we replace it with a variable~$s_{\text{new}}\in\mathbb{R}$ to form~$\max(\rho_{\text{rev}}^{\varphi_1}, ..., s_{\text{new}}, ..., \rho_{\text{rev}}^{\varphi_r})$. We then introduce the equality constraint~$\rho_{\text{rev}}^{\varphi_i} = s_{\text{new}}$, which results in form of either~\eqref{eq:basic:max} or~\eqref{eq:basic:min}. A similar replacement is applied for the~$\min$-type constraint in~\eqref{eq:basic:min}. These replacements are applied to all non-predicate arguments repeatedly until the resulting constraints take the form of~\eqref{eq:basic:max} or~\eqref{eq:basic:min}, with every argument of the~$\max$ and~$\min$ functions being either a predicate function (or introduced variables, such as~$s_{\text{new}}$). These equality constraints are then encoded as two inequalities, as described earlier for equality~\eqref{eq:dc3}.  

The resulting program is a DC program, derived using the \textit{composition rule}~\cite[Section 3.2.4]{Boyd2004-kn} and the \textit{monotonicity} properties of the~$\max$ and~$\min$ functions. Specifically, as the~$\max$ function is increasing, its composite function is convex when all arguments are convex. Similarly, as the~$\min$ function is increasing, its composite function is concave when all arguments are concave. Assumptions~\ref{assum:dcpredicate} and~\ref{assum:maxmin} guarantee these conditions by ensuring that the arguments of these functions are either predicate functions or introduced variables. Therefore, all composite functions in~\eqref{eq:basic:max} and~\eqref{eq:basic:min} are either convex or concave, completing the proof.
\end{proof} 
The proof above demonstrates a basic procedure for reformulating Problem~\ref{problem} into a DC problem. While this procedure effectively illustrates the fundamental idea of robustness decomposition, it relies on Assumption~\ref{assum:maxmin} and does not fully exploit the problem's inherent structure. Specifically, it increases the number of concave constraints unnecessarily, which can result in more iterations or suboptimal local solutions. This issue arises because the new constraints of the form~\eqref{eq:basic:max} are expressed as \textit{equalities} and are then encoded as two inequalities in CCP. Due to the nonaffine nature, even the convex~$\max$-type constraint in \eqref{eq:basic:max} introduces a concave inequality, and more for the concave~$\min$-type constraint in \eqref{eq:basic:min}. These limitations motivate the introduction of a refined version of the reformulation procedure in the subsequent subsections, which addresses these drawbacks by reducing the non-convex components of the resulting program.

\begin{remark}
The decomposition of~$\rho_{\text{rev}}^{\varphi}$ corresponds to a decomposition of the robustness tree~$\mathcal{T}^{\varphi}$ from top to bottom, in accordance with the tree structure.
\end{remark}

\subsection{Fewer non-convex parts via epigraphic reformulations}\label{subsec:fewer}

The refined procedure consists of two parts. The first part involves decomposing the robustness function in the objective into a series of constraints. The second part involves recursively decomposing these constraints into simpler ones.

\begin{procedure}\label{procedure} \textbf{(Part i. From cost to constraints)} The introduced constraints vary depending on whether the outermost operator of the robustness function~$\rho_{\text{rev}}^{\varphi}$ is a~$\max$ or~$\min$.

\textbf{(i. max)} For when the outermost operator of~$\rho_{\text{rev}}^{\varphi}$ is max, i.e.,~$\rho_{\text{rev}}^{\varphi} = \max(\rho_{\text{rev}}^{\varphi_1},\rho_{\text{rev}}^{\varphi_2},...,\rho_{\text{rev}}^{\varphi_r})$, we introduce a new variable~$s_\xi\in\mathbb{R}$, and reformulate the program as follows:
\begin{subequations}\label{eq:max_trans}
\begin{align}
\min _{\boldsymbol{x}, \boldsymbol{u},s_\xi}\hspace{0.2em} &s_\xi \label{eq:min_xi}\\
\text { s.t. } &~\eqref{eq:min_sys1} ,\eqref{eq:min_xu1}\label{eq:x}\\
& s_\xi \leq 0   \label{eq:xi_0}\\
& \rho_{\text{rev}}^{\varphi_1}\left(\boldsymbol{x}\right) \leq s_\xi \dots \rho_{\text{rev}}^{\varphi_r}\left(\boldsymbol{x}\right)
\leq s_\xi \label{eq:xl}
\end{align}
\end{subequations}

\textbf{(i. min)} For when the outermost operator of~$\rho_{\text{rev}}^{\varphi}$ is~$\min$, that is,~$\rho_{\text{rev}}^{\varphi} = \min(\rho_{\text{rev}}^{\varphi_1},\rho_{\text{rev}}^{\varphi_2},...,\rho_{\text{rev}}^{\varphi_r})$, we reformulate the program similarly to~\eqref{eq:max_trans} using a new variable~$s_\xi$. However, rather than the constraint~\eqref{eq:xl}, we have 
\begin{align}\label{eq:xll}
\min (\rho_{\text{rev}}^{\varphi_1} ,\dots ,\rho_{\text{rev}}^{\varphi_r})\leq s_\xi  
\end{align}
Before we describe the second part of the procedure, the following lemma shows the equivalence of \textbf{(Part i)}, regardless of the type of the outermost operator.
\begin{lemma} \label{lem;eqcost}
The procedure in \textbf{(Part i)} equivalently transforms the original program~\eqref{eq:moto}.
\end{lemma}
\begin{proof}
The set of inequalities~\eqref{eq:xl} is equivalent to 
\begin{equation}\label{eq:eq1}
\max(\rho_{\text{rev}}^{\varphi_1},\rho_{\text{rev}}^{\varphi_2},...,\rho_{\text{rev}}^{\varphi_r}) \leq s_\xi.
\end{equation}
We can use a similiar proof of \cite[Proposition 5]{stlccp} (see \cite[Section 8.3.4.4]{El_Ghaoui2015-eb}). 
The only difference is the nonconvex terms in the program, that is,~\eqref{eq:min_xu1} and~\eqref{eq:min_sys1}, which does not affect the proof. In particular, the equality condition of the inequality conditions
in~\eqref{eq:xll} and~\eqref{eq:eq1} are satisfied when the cost function~$s_\xi$ takes a local minimum.
\end{proof}

\textbf{(Part ii. from constraints to constraints)} After \textbf{(Part i)}, regardless of the outermost operator of the robustness function, all the constraints in~\eqref{eq:xl} and~\eqref{eq:xll} above are in either of the two forms~\eqref{eq:max_before} or~\eqref{eq:mellowmin}:
\begin{align}
\max(\rho_{\text{rev}}^{\Phi_1^{(1)}},...,\rho_{\text{rev}}^{\Phi_{y_{\max}}^{(1)}})\leq s_{\max}^{(1)}, \label{eq:max_before}\\
\min(\rho_{\text{rev}}^{\Psi_1^{(1)}},...,\rho_{\text{rev}}^{\Psi_{y_{\min}}^{(1)}}) \leq s_{\min}^{(1)}, \label{eq:mellowmin} 
\end{align}
where~$s_{\min}^{(1)},s_{\max}^{(1)}\in\mathbb{R}$ are variables and~$\rho_{\text{rev}}^{\Phi_i^{(1)}}$ and~$\rho_{\text{rev}}^{\Psi_i^{(1)}}$ are robustness functions corresponding to the subformulas~$\Phi_i^{(1)}$ and~$\Psi_i^{(1)}$, for~$i = \{1, \ldots, y_{\max}^{(1)}\}$ and~$i = \{1, \ldots, y_{\min}^{(1)}\}$, respectively.
These inequalities are paraphrased expressions of~\eqref{eq:xl} and~\eqref{eq:xll} above.
We continue the following steps until all the arguments in the~$\max$ and~$\min$ functions become predicate functions.

\textbf{(ii. max)} For inequality constraints of the form~\eqref{eq:max_before}, we transform it as follows:
\begin{equation}\label{eq:max_after}
\begin{aligned}[b]
\rho_{\text{rev}}^{\Phi_1^{(1)}}\leq s_{\max}^{(1)},\rho_{\text{rev}}^{\Phi_2^{(1)}}\leq s_{\max}^{(1)} ,..., \rho_{\text{rev}}^{\Phi_{y_{\max}}^{(1)}}\leq s_{\max}^{(1)}.
\end{aligned}
\end{equation}
This is, of course, equivalent to~\eqref{eq:max_before}.

\textbf{(ii. min)} For constraints of the form~\eqref{eq:mellowmin}, we check whether each argument of the~$\min$ of the left-hand side is a predicate function or not. We then replace arguments that are not a predicate function, say~$\rho_{\text{rev}}^{\Psi_i^{(1)}}$, by a variable~$s_{\text{new}}\in\mathbb{R}$ as 

\begin{equation}\label{eq:min_slack}
\min(\rho_{\text{rev}}^{\Psi_1^{(1)}},...,s_{\text{new}},...,\rho_{\text{rev}}^{\Psi_{y_{\min}}^{(1)}})\leq s_{\min}^{(1)}.
\end{equation}
Then, we add the following constraints:
\begin{equation}\label{eq:ii_minmax}
\rho_{\text{rev}}^{\Psi_i^{(1)}} \leq s_{\text{new}}, 
\end{equation}


After the procedures in \textbf{(Part ii)}, each constraint in \eqref{eq:max_after}--\eqref{eq:ii_minmax} takes the form of either \eqref{eq:max_before} or \eqref{eq:mellowmin}, with updated arguments and variables. Consequently, we apply the procedure in \textbf{(Part ii)} to each of these constraints again. For instance, consider the next constraint of the form \eqref{eq:max_before}. Denote its arguments as~$\rho_{\text{rev}}^{\Phi_i^{(2)}},i=\{1,...,y_{\max}^{(2)}\}$ and its variables as~$s_{\max}^{(2)}$, For this constraint, we follow the step in \textbf{(ii. max)}. This process is repeated iteratively for each constraint until we reach the predicate functions, the bottom of the STL tree.
This concludes the procedure. 
\end{procedure}

The resulting program, which is denoted as~$\mathcal{P}_{\text{DC}}$, is written as follows:
\begin{subequations}\label{eq:final}
\begin{align}
        \min _{\boldsymbol{z}}
        \hspace{0.2em}&s_\xi \label{eq:finalcost}\\
        \text { s.t. } &~\eqref{eq:min_sys1} ,\eqref{eq:min_xu1} \label{eq:final_xuu}\\
        & s_\xi \leq 0,  \label{eq:final_xii} \\
&
\left.\hspace{-1cm}
\begin{tabular}{l}
$\rho_{\text{rev}}^{\Phi_1}\leq s^{(1)}_{\max},\rho_{\text{rev}}^{\Phi_2}\leq s^{(1)}_{\max},...,\rho_{\text{rev}}^{\Phi_i}\leq s^{(1)}_{\max}$,  \\
$\rho_{\text{rev}}^{\Phi_{i+1}}\leq s^{(2)}_{\max},\rho_{\text{rev}}^{\Phi_{i+2}}\leq s^{(2)}_{\max},...,\rho_{\text{rev}}^{\Phi_j}\leq s^{(2)}_{\max}$,  \\
$\quad\quad\vdots\quad\quad\hspace{1cm}\vdots \hspace{2.5cm}~$  \\
$\rho_{\text{rev}}^{\Phi_k}\leq s^{(v)}_{\max},\rho_{\text{rev}}^{\Phi_{k+1}}\leq s^{(v)}_{\max},...,\rho_{\text{rev}}^{\Phi_{l}}\leq s^{(v)}_{\max}$  \\
\end{tabular}
\right\} \begin{tabular}{l}\text{from }$\max$\\
\text{nodes} \end{tabular}\label{eq:finalmax}\\
&
\left.\hspace{-1cm}
\begin{tabular}{l}
$\min(\rho_{\text{rev}}^{\Psi^{(1)}_1},...,\rho_{\text{rev}}^{\Psi^{(1)}_{y_{\min}}}) \leq s^{(1)}_{\min}$, \\
$\quad\quad\vdots$ \\
$\min(\rho_{\text{rev}}^{\Psi^{(w)}_1},...,\rho_{\text{rev}}^{\Psi^{(w)}_{y_{\min}}}) \leq s^{(w)}_{\min}~$,\\
\end{tabular}
\right\}  \begin{tabular}{l}\text{from }$\min$\\
\text{nodes} \end{tabular}\label{eq:finalmin}
\end{align}
\end{subequations}
where~$\boldsymbol{z}=[\boldsymbol{x}^{\mathsf{T}},\boldsymbol{u}^{\mathsf{T}},s_\xi,s^{(1)}_{\max},...,s^{(v)}_{\max},s^{(1)}_{\min},...,s^{(w)}_{\min}]^{\mathsf{T}}\in \mathbb{R}^{(n+m)H+1+v+w}$ denotes the vector of optimization variables, and~$l$ and~$w$ represent the number of constraints in~\eqref{eq:finalmax} and~\eqref{eq:finalmin}, respectively. There is a slight abuse of notation: either~$s_{\max}^{(1)}$ or~$s_{\min}^{(1)}$ corresponds to~$s_\xi$, while some arguments of the~$\min$ functions in \eqref{eq:finalmin} are introduced variables~$s^{(\boldsymbol{\cdot})}_{\max}$. 
All the formulas in this program, i.e.,~$\Phi_1, \ldots, \Phi_l$ and~${\Psi^{(i)}_1}, \ldots, {\Psi^{(i)}_{y_{\min}}}$ (for~$i \in \{1, \ldots, w\}$), are predicates, and the program no longer contains any~$\max$ functions.
Although~$w$ is equal to the number of times we apply the procedure in \textbf{(ii. min)} to constraints of the form~\eqref{eq:mellowmin},~$l$ is not necessarily equal to $v$ nor to~$\sum_{i=1}^{v} y_{\max}^{(i)}$.

\begin{remark}
Procedure~\ref{procedure} above uses the idea of the epigraphic reformulations \cite[Sec.~4.1]{Boyd2004-kn}, which transform the convex (typically linear) cost function into the corresponding epigraph constraints. However, our approach is different from the traditional epigraphic reformulation-based approach in the sense that we handle \textit{nonconvex} cost functions and we have to apply reformulation \textit{recursively}.
\end{remark}

\subsection{Removing successive~\(\min\) by formula simplification}\label{subsection:simplification}

Although Procedure~\ref{procedure} is theoretically equivalent to the original program (will be proven in Theorem~\ref{theo:sameoptimal}), the inequalities~\eqref{eq:min_slack} and~\eqref{eq:ii_minmax} differ from the original inequality~\eqref{eq:mellowmin}. Specifically, during the optimization process, if the variable~$s_{\text{new}}$ is not the minimum among the arguments of the~$\min$ functions in~\eqref{eq:min_slack}, the equality condition in~\eqref{eq:ii_minmax} may not hold. Consequently, a discrepancy can occur between the left-hand side and the right-hand side of~\eqref{eq:mellowmin}.
Although the transformation is theoretically valid, its practical application requires refinement to mitigate such discrepancies. To address this issue, we adopt a simplification technique inspired by \cite{Kurtz2022-pe} and apply it to the robustness tree. This involves equivalently transforming the formulas by recursively simplifying consecutive nodes of the same type (either~$\max$ or~$\min$), starting from the root and proceeding downward, until no such consecutive nodes remain. By reducing the number of~$\min$-type nodes, this simplification effectively minimizes potential practical discrepancies between the two programs.
\begin{example}
Fig.~\ref{fig:simplified} illustrates this technique applied to the robustness tree in Fig.~\ref{fig:robustnesstree}. The number of subtrees whose top nodes are~$\min$ type,~$N_{\vee}^\varphi$, was reduced from~$3$ to~$2$.
\end{example}

\begin{figure}[ht]
    \centering
    \includegraphics[keepaspectratio, scale=0.5]{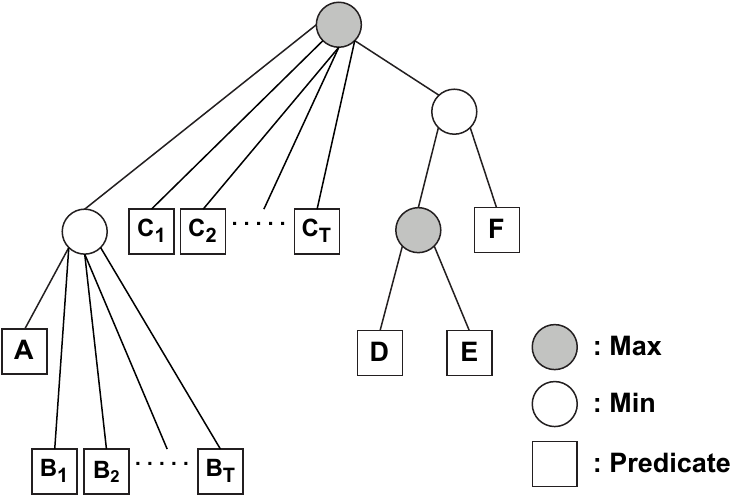}
    \caption{The simplified version of the robustness tree in \protect{Fig.~\ref{fig:robustnesstree}}.}
    \label{fig:simplified}
\end{figure}

In the rest of this paper, formula~$\varphi$ refers to its simplified version. Importantly, the robustness simplification does not affect Procedure~\ref{procedure} because it is performed offline and does not alter the type~$\mathcal{A}$ of the tree~$\mathcal{T}^\varphi$ or its subtrees. The overall reformulation procedure in Sections~\ref{subsec:fewer} and~\ref{subsection:simplification} is summarized in Algorithm~\ref{alg:reformulation}.

\begin{algorithm}
\small
\caption{Program reformulation}\label{alg:reformulation}
\begin{algorithmic}[1]
	\Require{Original program~\eqref{eq:moto}}
	\Ensure{The reformulated program (\textbf{Cost},\textbf{Constr})}
        \State{Simplify the formula;}
        \State{\textbf{Constr}$\gets\{\eqref{eq:min_sys1},~\eqref{eq:min_xu1},\eqref{eq:xi_0}\}$; \textbf{Cost}$\gets\{\eqref{eq:min_xi}\}$};
        \If{$\rho_{\text{rev}}^{\varphi}=\max \left(\rho_{\text{rev}}^{\varphi_1}, \rho_{\text{rev}}^{\varphi_2}, \ldots, \rho_{\text{rev}}^{\varphi_r}\right)$}
        \State{\textbf{Constr}$\gets\textbf{Constr}\cup\{\eqref{eq:xl}\}$;}
        \Else{ (if~$\rho_{\text{rev}}^{\varphi}=\min \left(\rho_{\text{rev}}^{\varphi_1}, \rho_{\text{rev}}^{\varphi_2}, \ldots, \rho_{\text{rev}}^{\varphi_r}\right)$)}
                    \State{\textbf{Constr}$\gets\textbf{Constr}\cup\eqref{eq:xll}\}$};
        \EndIf

        \While{$\neg(\forall \text{constr}\in\textbf{Constr}\text{, every argument is either a predicate}\linebreak\text{function or an introduced variable})$}
            \If{constr is max form~\eqref{eq:max_before}}
                \ForEach{$i\in\{1,...,y_{\max}\}$  }
                    \If{not all the arguments are predicates}
                    \State{\textbf{Constr}$\gets(\textbf{Constr}-\{\eqref{eq:max_before}\} )\cup\{\eqref{eq:max_after}\}$;}

                    \EndIf
                \EndFor
            \EndIf
        
            \If{constr is min form~\eqref{eq:mellowmin}}
                    \ForEach{$i\in\{1,...,y_{\min}\}$ in~\eqref{eq:mellowmin}}
                    \State{\textbf{Constr}$\gets(\textbf{Constr}-\{\eqref{eq:mellowmin}\})\cup\{\eqref{eq:min_slack}\}\cup\{\eqref{eq:ii_minmax}\}$;}

                \EndFor
            \EndIf
        \EndWhile
        
\end{algorithmic}
\end{algorithm}

\subsection{Equivalence analysis}\label{subsec:result}

The following theorem first clarifies that Procedure~\ref{procedure} retains the soundness property of the robustness function, ensuring that program~\eqref{eq:final} guarantees formula satisfaction.
\begin{theo}\label{theo:satisfy}
A feasible solution~$\boldsymbol{z}$ of program~\eqref{eq:final} always satisfies the STL specification~$\varphi$, i.e.,~$\boldsymbol{x} \vDash \varphi$.
\end{theo}
\begin{proof}
We can focus on \textbf{(Part ii)} because of Lemma~\ref{lem;eqcost}. Consider the initial replacement of the robustness function~$\rho_{\text{rev}}^{\Psi_i}(\boldsymbol{x})$ with a new variable~$s_{\text{new}}$. This transformation converts inequality~\eqref{eq:mellowmin} into two inequalities,~\eqref{eq:min_slack} and~\eqref{eq:ii_minmax}. Because~$\min$ is an increasing function with respect to each argument, we have inequality 
$
\min(\rho_{\text{rev}}^{\Psi_1^{(1)}},...,\rho_{\text{rev}}^{\Psi_i^{(1)}},...,\rho_{\text{rev}}^{\Psi_{y_{\min}}^{(1)}}) \leq \min(\rho_{\text{rev}}^{\Psi_1^{(1)}},...,s_{\text{new}},...,\rho_{\text{rev}}^{\Psi_{y_{\min}}^{(1)}}).
$
Using this inequality with~\eqref{eq:min_slack}, inequality~\eqref{eq:mellowmin} holds even after the transformation involving~$s_{\text{new}}$. A similar reasoning applies to all subsequent transformations.
Thus, the original inequalities in~\eqref{eq:xl} eventually hold. Combining this with~$s_\xi \leq 0$,~$\rho_{\text{rev}}^{\varphi}(\boldsymbol{x}) \leq 0$ holds. From~\eqref{eq:sound}, the statement is proven.
\end{proof}

The converse of Theorem~\ref{theo:satisfy} also holds. Let~$\Phi^{(i)}_{\text{parent}}$ denote the formula for the parent tree of its first-level subtrees~$\Phi^{(i)}_1, \dots, \Phi^{(i)}_{y_{\max}}$ for~$i \in \{1, \dots, l\}$, and let~$\Psi^{(i)}_{\text{parent}}$ be defined similarly for~$\Psi^{(i)}_1, \dots, \Psi^{(i)}_{y_{\max}}$ with~$i \in \{1, \dots, w\}$.

\begin{lemma}\label{lem:forward_feasible}
Let~$(\boldsymbol{x}^*,\boldsymbol{u}^*)$ denote a feasible solution of~\eqref{eq:moto}, then the following solution~$\boldsymbol{z}^*$ is feasible for~\eqref{eq:final}.
\begin{align}\label{z_star}
&\boldsymbol{z}^*:=(\boldsymbol{x}^*,\boldsymbol{u}^*,
s_\xi^*=\rho_{\text{rev}}^{\varphi}(\boldsymbol{x}^*),
s_{\max}^*=(\rho_{\text{rev}}^{\Phi^{(1)}_{\text{parent}}}(\boldsymbol{x}^*),\nonumber \\ 
&...,\rho_{\text{rev}}^{\Phi^{(v)}_{\text{parent}}}(\boldsymbol{x}^*)),s_{\min}^*=(\rho_{\text{rev}}^{\Psi^{(1)}_{\text{parent}}}(\boldsymbol{x}^*),...,\rho_{\text{rev}}^{\Psi^{(w)}_{\text{parent}}}(\boldsymbol{x}^*)))
\end{align}
\end{lemma}
\begin{proof}
Using the substitution~$\boldsymbol{z}^*$ from~$(\boldsymbol{x}, \boldsymbol{u}, s_\xi, s_{\max}, s_{\min})$ to~$(\boldsymbol{x}, \boldsymbol{u})$ in \eqref{z_star}, the equality condition of all inequalities in~\eqref{eq:finalmax} and~\eqref{eq:finalmin} is satisfied. Consequently, the feasible region of~\eqref{eq:final} becomes identical to that of~\eqref{eq:moto}.  
\end{proof}

Finally, program~\eqref{eq:final} is \textit{equivalent} to the original program~\eqref{eq:moto}, and it is a DC program under Assumptions~\ref{assum:dc}--\ref{assum:dcpredicate}. This reformulation does not require Assumption~\ref{assum:maxmin}, which was required in the procedure outlined in the proof of Proposition~\ref{theo:dcprogram}.
\begin{theo}\label{theo:sameoptimal} Procedure~\ref{procedure} can equivalently transform program~\eqref{eq:moto} into DC program~$\mathcal{P}_{\text{DC}}$~\eqref{eq:final} under Assumptions~\ref{assum:dc}--\ref{assum:dcpredicate}.
\end{theo}
\begin{proof} 
Let a feasible solution ($\boldsymbol{x}^*,\boldsymbol{u}^*$) be optimal for~\eqref{eq:moto}, then we first show~$\boldsymbol{z}^*$ in~\eqref{z_star} is optimal for~\eqref{eq:final}. From Lemma~\ref{lem:forward_feasible}, if~$\boldsymbol{z}^*$ is not the optimal solution for~\eqref{eq:final}, there exists another feasible solution~$(\boldsymbol{x}',\boldsymbol{u}',s_\xi', s_{\max}',s_{\min}')$, with a better objective value, i.e., such that~$s_\xi'< s_\xi^*=\rho_{\text{rev}}^{\varphi}(\boldsymbol{x}^*)$. As in the proof of Theorem~\ref{theo:satisfy},~$\rho_{\text{rev}}^{\varphi}(\boldsymbol{x}')\leq s_\xi'$. Therefore,~$\rho_{\text{rev}}^{\varphi}(\boldsymbol{x}')<\rho_{\text{rev}}^{\varphi}(\boldsymbol{x}^*)$. However, as Theorem~\ref{theo:satisfy} shows that~$\boldsymbol{x}'$ is also feasible for~\eqref{eq:moto}, this inequality contradicts the fact that~$\boldsymbol{x}^*$ is optimal for~\eqref{eq:moto}. Therefore,~$\boldsymbol{z}^*$ is optimal for~\eqref{eq:final} and the global optimum of both programs~\eqref{eq:moto} and~\eqref{eq:final} is~$\rho_{\text{rev}}^{\varphi}(\boldsymbol{x}^*)$. We can also prove analogously that if a feasible solution~$\boldsymbol{z}$ is optimal for~\eqref{eq:final}, then~$(\boldsymbol{x},\boldsymbol{u})$ is optimal for~\eqref{eq:moto}. 

Next, the resulting program \eqref{eq:final} is confirmed as a DC program following the composition rule~\cite[Section 3.2.4]{Boyd2004-kn} and the monotonicity of the $\min$ function, similar to the proof of Proposition~\ref{theo:dcprogram}. However, unlike Proposition~\ref{theo:dcprogram}, Assumption~\ref{assum:maxmin} on predicate functions is not required, as the $\max$ functions are all decomposed in Algorithm~\ref{alg:reformulation}.
\end{proof}

\subsection{Program complexity under more restrictive assumptions}\label{subsec:effresult}
This subsection demonstrates the efficiency of Procedure~\ref{procedure} when applied with CCP under more restrictive assumptions. To illustrate the structure of the program~$\mathcal{P}_{\text{DC}}$, consider the following key property:
\begin{prop}\label{prop:onetoone}
The number of constraints in~\eqref{eq:finalmax} corresponds to the number of predicates whose parent nodes are~$\max$-type, that is,~$l=N^{\varphi}_{p\wedge}$. 
Similarly, the number of constraints in~\eqref{eq:finalmin} corresponds to the number of~$\min$-type subtrees of the robustness tree~$\mathcal{T}^\varphi$, that is,~$w=N^\varphi_{\vee}$. 
Each constraint, indexed by~$i \in \{1, \ldots, N_{p\wedge}^\varphi\}$ (respectively~$i \in \{1, \ldots, N_\vee^\varphi\}$), matches uniquely with a~$\max$-type leaf~$\mathcal{L}^\varphi_{\wedge,i}$ (respectively a~$\min$-type subtree~$\mathcal{T}^\varphi_{\vee,i}$).
\end{prop}
\begin{example}
Furthermore, consider the formula in Fig.~\ref{fig:robustnesstree}. Algorithm~\ref{alg:reformulation} produces~$N_{p\wedge}^\varphi=T+2$ constraints of the form~\eqref{eq:finalmax}:
\begin{align}\label{eq:maxmax}
g^{C_1} \leq s_\xi, \dots, g^{C_T} \leq s_\xi, g^D \leq s_{\max}, g^E \leq s_{\max},
\end{align}
and ~$N_{\vee}^\varphi=2$ constraints of the form~\eqref{eq:finalmin}:
\begin{align}\label{eq:minmin}
\hspace{-0.25cm}\min(g^A, g^{B_1}, \dots, g^{B_T}) \leq s_{\min}^{(1)},\min(s_{\max}, g^F) \leq s_{\min}^{(2)},
\end{align}
where~$s_\xi$, $s_{\max}$, and~$s_{\min}=(s_{\min}^{(1)},s_{\min}^{(2)})$ are the variables introduced during the robustness decomposition procedure as those in~\eqref{eq:final}. Each of these constraints represent a~$\max$-type leaf~$\mathcal{L}^\varphi_{\wedge}= \{g^{C_1}, \dots, g^{C_T}, g^D, g^E\}$ and a~$\min$-type subtree~$\mathcal{T}^\varphi_{\vee}= \{\mathcal{T}^{(A\vee\Diamond_{[1, T]}B)}, \mathcal{T}^{(D\wedge E)\vee F}\}$, respectively.
\end{example}

Leveraging this property, the following theorem demonstrates the effectiveness of the proposed robustness decomposition under the assumption required for its validity.
\begin{assumption}\label{assum:major}
\textit{(System and state/input constraints)} The system~\eqref{eq:system} is linear, i.e.,~$x_{t+1} =A x_t+B u_t$, where~$A\in\mathbb{R}^{n\times n}$ and~$B\in\mathbb{R}^{n\times m}$. The functions~$h_x,h_u$ of~$\mathcal{X},\mathcal{U}$ are all convex.
\end{assumption}
\begin{theo}\label{theo:major}
Under Assumptions~\ref{assum:dcpredicate} and~\ref{assum:major}, CCP only majorizes inequalities of the form~\eqref{eq:finalmin}, which represents~$\min$-type subtrees~$\mathcal{T}^\varphi_{\vee,i}, i \in \{1, \ldots, N_\vee^\varphi\}$.
\end{theo}
\begin{proof}
The program~$\mathcal{P}_{\text{DC}}$ does not have nonaffine equality constraints~\eqref{eq:dc3} as the system is restricted to linear. The program~$\mathcal{P}_{\text{DC}}$ is a convex program except for the concave constraints of the form~\eqref{eq:finalmin}.
\end{proof}
In other words, the concave parts in~$\mathcal{P}_{\text{DC}}$ originate solely from the disjunctive operators in the simplified formula~$\varphi$, thereby minimizing the approximation in the majorization step.

Moreover, using the proposed decomposition, the subproblem can be a linear program if we impose stronger affine assumptions on sets~$\mathcal{X}, \mathcal{U}$ and STL predicate functions:
\begin{assumption}\label{assum:linear}
\textit{(System, state/input constraints, and STL predicate functions)} The system~\eqref{eq:system} is linear, i.e.,~$x_{t+1} =A x_t+B u_t$. The functions~$h_x,h_u$ of~$\mathcal{X},\mathcal{U}$, as well as predicate functions~$g^{\mu_i}\in \mathcal{M}$ whose parent node is~$\max$-type, are all affine (with respect to their respective arguments). However, predicate functions~$g^{\mu_i}$ whose parent node is~$\min$-type are concave.
\end{assumption}
\begin{theo}\label{theo:convexquadratic}
Under Assumption~\ref{assum:linear}, the program~$\mathcal{P}_{\text{DC}}$ becomes a linear program after the majorization. If we add the quadratic cost function, it becomes a quadratic program.
\end{theo}
\begin{proof}
Similar to Theorem~\ref{theo:major}, the result follows by linearizing the concave constraints~\eqref{eq:finalmin} in the program~$\mathcal{P}_{\text{DC}}$. Procedure~\ref{procedure} does not produce nonaffine equality constraints.
\end{proof}
Thus, under Assumption~\ref{assum:linear}, the proposed method reduces to sequential linear programming. If we add a quadratic cost function in~\eqref{eq:min_x1}, it becomes sequential quadratic programming (SQP). This is crucial since under general conditions in Assumptions~\ref{assum:dc} and~\ref{assum:dcpredicate}, the subproblem becomes a general convex program, solved sequentially via CCP. We denote the convex subproblem of program~$\mathcal{P}_{\text{DC}}$ after the majorization as~$\mathcal{P}_{\text{CP}}$.

\section{Iterative Optimization with Penalty CCP}\label{section:CCP}
This section explores the properties of the optimization step under two different CCP settings: the standard CCP and the tree-weighted penalty CCP.

\subsection{Standard CCP}\label{subsec:subproblem}

The standard CCP described in Section~\ref{subsec:ccp} guarantees formula satisfaction. Let~$\boldsymbol{z}=[\boldsymbol{x}^{\mathsf{T}},\ldots]^\mathsf{T}$ be a feasible solution to the resulting program~$\mathcal{P}_{\text{CP}}$. 
\begin{theo}\label{theo:lp_satisfy}
A feasible solution~$\boldsymbol{z}$ of program~$\mathcal{P}_{\text{CP}}$ always satisfies the STL specification~$\varphi$, that is,~$\boldsymbol{x} \vDash \varphi$.
\end{theo}
\begin{proof}
The first-order approximations of CCP at each step are global over-estimators. Specifically, for any DC function~$ p$ in Definition~\ref{def:dcf}, it holds that $p(\boldsymbol{z})=q(\boldsymbol{z})- r(\boldsymbol{z})\leq q(\boldsymbol{z})-r(\boldsymbol{z}_{(i)})+ \nabla r(\boldsymbol{z}_{(i)})^{\mathsf{T}}(\boldsymbol{z}_{(i)}-\boldsymbol{z})$ where~$\boldsymbol{z}_{(i)}$ is the current iteration point of variable~$\boldsymbol{z}$.
Therefore, a feasible solution~$\boldsymbol{z}$ of program~$\mathcal{P}_{\text{CP}}$ is a feasible solution to program~$\mathcal{P}_{\text{DC}}$.
By Theorem~\ref{theo:satisfy}, the statement follows.
\end{proof}

Moreover, due to this global nature of the inequality bounds in the proof of Theorem~\ref{theo:lp_satisfy} above, standard CCP has the following useful three properties (see~\cite{Lipp2016-fa,Sriperumbudur_lanckriet_2009} for a proof):
\begin{itemize}\label{item:eqconst}
\item[\rnum{1})] All of the iterates are feasible.
\item[\rnum{2})] The cost value converges. 
\item[\rnum{3})] CCP does not need to restrict the update at each iteration nor to perform a line search.
\end{itemize}
In particular, the third point differs from traditional SQP methods that often constrain the update within trust regions. However, these statements are formally ensured only if the starting point is in the feasible set, which is one disadvantage of standard CCP. 


\subsection{Tree-weighted penalty CCP (TWP-CCP)}\label{subsec:twp-ccp}

The TWP-CCP improves the algorithm by exploiting STL's hierarchical information. In an STL specification, not all logical operators hold equal importance. Nodes with more leaves are considered more critical because their precision significantly affects the overall tree structure and performance. To incorporate this priority information, we extend the idea of \textit{penalty CCP}~\cite{Lipp2016-fa}, which relaxes constraints by adding positive variables~$s_j\in\mathbb{R}$ to each constraint $j$ and penalizing their sum in the objective function. Specifically, let the subtree of~$\mathcal{T}^{\varphi}$ associated with the constraint $j\in \{1, \ldots, m\}$ be denoted as~$\mathcal{T}^{\varphi_j}$. We then assign priority weights to these penalty variables~$s_j$ based on the number of leaves~$N_p^{\varphi_j}$ the robustness tree~$\mathcal{T}^{\varphi_j}$ has. 
The relaxed version of program~\eqref{eq:dc} is:
\begin{subequations}\label{eq:leafnodeweighted}
\begin{align}
\min _{\boldsymbol{z}} \hspace{0.5em}& p_0(\boldsymbol{z})+\tau_{(i)} \sum_{j=1}^m N_p^{\varphi_j} s_j \label{eq:relaxobj}\\
\text { s.t. } 
& p_j\left(\boldsymbol{z}\right) \leq s_j \text{ and } s_j \geq 0, \quad j=\{1, \ldots, m\}. \label{eq:relaxcons}
\end{align}
\end{subequations}
In principle, we can introduce the penalty variables~$s_j$ in any constraint. For the program~$\mathcal{P}_{\text{DC}}$, the penalty variables were introduced only for~$w=N_\vee^{\varphi}$ concave constraints~\eqref{eq:finalmin}. We denote the relaxed version of the program~$\mathcal{P}_{\text{DC}}$ as~$\mathcal{P}_{\text{DC}}^r$ and its convex subproblem, obtained after the majorization, as $\mathcal{P}_{\text{CP}}^r$.

The TWP-CCP is summarized in Algorithm~\ref{alg:ccp}, where~$\boldsymbol{z}_{(i)}$ represents the value of the variables~$\boldsymbol{z}$ at iteration step~$i$. The cost function is defined as~$c(\boldsymbol{z}) := p_0(\boldsymbol{z}) + \tau_{(i)} \sum_{j=1}^m N_p^{\varphi_j} s_j$, and its value at iteration step~$i$ is denoted by~$c_{(i)} := c(\boldsymbol{z}_{(i)})$.
Other symbols mean the same as in Section~\ref{subsec:ccp}.

\begin{algorithm}
\caption{Tree-Weighted Penalty CCP}\label{alg:ccp}
\begin{algorithmic}[1]
	\Require{Parameters: an initial point~$z_0$, an initial weight~$\tau_0\in\mathbb{R}_{>0}$ and its maximum constant~$\tau_{\max}>\tau_0$ and its rate of change at each iteration,~$\nu>1$, the number of leaves~$N_p^{\varphi_j}$ associated with each constraint for~$j=\{1, \ldots, m\}$, and constants~$s_{\text{c}}\in\mathbb{R}_{\geqslant 0},s_{\text{sp}}\in\mathbb{R}_{\geqslant0}$;}
    \State{Compute the number of leaves~$N_p^{\varphi_j}$ associated with each constraint for~$j=1,...,m$, and initialize~$i:=0$;}
	\State{Majorize the concave terms~$-r_j,j=0, \ldots, m$ in~\eqref{eq:leafnodeweighted} to~$-\hat{r}_j$ where~$\hat{r}_j\left(z\right)=r_j\left(z_{(i)}\right)+\nabla r_j\left(z_{(i)}\right)^\mathsf{T}\left(z-z_{(i)}\right)$ and solve the convex program; }
	\State{Update~$\tau_{(i+1)}:=\min \left(\nu \tau_{(i)}, \tau_{\max }\right)$;}
	\State{Update iteration as~$i:=i+1$. Repeat until the termination conditions are met:~$\max_j(s_j) \leq s_{\text{c}}$ for penalty variables~$s_j$ and~$c_{(i+1)} - c_{(i)} \leq s_{\text{ep}}$ for cost function~$c$;}
\end{algorithmic}
\end{algorithm}
 
The termination condition is met when both the constraint violations and the improvement in cost are small, that is, 
\begin{align}\label{eq:stop}
\max_j(s_j)\leq s_{\text{c}}\text{ and }c_{(i+1)} - c_{(i)} \leq s_{\text{ep}},
\end{align}
where~$s_{\text{c}}$ and~$s_{\text{ep}}$ are small positive values close to zero. 

The TWP-CCP removes the requirement for an initial feasible solution in standard CCP, enabling the algorithm to begin with an initial point outside the feasible region of~$\mathcal{P}_{\text{CP}}$ while still ensuring formula satisfaction, as in Theorem~\ref{theo:lp_satisfy}. Additionally, the cost value also converges because, when~$\tau_{(i)} = \tau_{\max}$, the program~\eqref{eq:leafnodeweighted} with TWP-CCP behaves as a DC program applied with standard CCP.

The analysis of formula satisfaction in the TWP-CCP depends on the value of~$s_{\text{c}}$. When~$s_{\text{c}} = 0$, any feasible solution of the relaxed program~\eqref{eq:leafnodeweighted} lies within the feasible region of the original program~\eqref{eq:dc}, as the relaxed program converges to the original program by the end of the optimization, thereby ensuring formula satisfaction. Moreover, due to the monotonicity property, formula satisfaction can still be guaranteed for~$s_{\text{c}} \neq 0$, by appropriately adjusting the threshold in~\eqref{eq:min_xuu1}. 

\begin{theo}\label{theo:violate}
A feasible solution~$\boldsymbol{z}$ of program~$\mathcal{P}_{\text{CP}}^r$ with a modified threshold of~$-N^\varphi_{\vee} s_{\text{c}}$ (leading to~$s_\xi < -N^\varphi_{\vee} s_{\text{c}}$ rather than~\eqref{eq:final_xii}) satisfies the STL specification~$\varphi$, that is,~$\boldsymbol{x} \vDash \varphi$.
\end{theo}
\begin{proof}
Due to page limitations and for clarity, we sketch the proof of this theorem using the example of formula~$\varphi$ in Fig.~\ref{fig:robustnesstree} again.
First, a feasible solution $\boldsymbol{z}$ of program~$\mathcal{P}_{\text{CP}}^r$ is a feasible solution to program~$\mathcal{P}_{\text{DC}}^r$ similarly to the non-relaxed case in Theorem~\ref{theo:lp_satisfy}. Thus,  solution $\boldsymbol{z}$ satisfies \eqref{eq:maxmax} and the relaxed version of \eqref{eq:minmin} in~$\mathcal{P}_{\text{DC}}^r$.
By combining these constraints with~\eqref{eq:stop}, solution $\boldsymbol{z}$ satisfies~$T+2$ inequalities:
\begin{align}
    &\min(g^A, g^{B_1}, \dots, g^{B_T}) \leq s_\xi + s_{\text{c}}, \nonumber \\
    &\min(g^D, g^E, g^F) \leq s_\xi + s_{\text{c}}, g^{C_1} \leq s_\xi, \dots, g^{C_T} \leq s_\xi. \nonumber
\end{align}
By further combining these equations, $\boldsymbol{z}$ satisfies:
\begin{align}
    s_{\xi} &\geq \max\left(\min(g^A, g^{B_1}, \dots, g^{B_T}) - s_{\text{c}},  \min(g^D, g^E, g^F) - s_{\text{c}}, g^{C_1}, \dots, g^{C_T} \right)\nonumber \\
    &\geq \max\left(\min(g^A, g^{B_1}, \dots, g^{B_T}), \min(g^D, g^E, g^F),   g^{C_1}, \dots, g^{C_T} \right) - 2s_{\text{c}} \nonumber \\ 
    & = \rho_{\text{rev}}^{\varphi} - 2s_{\text{c}}. \label{eq:exfinal}
\end{align}
Since~$N^\varphi_{\vee} = 2$, \eqref{eq:exfinal} implies that~$\rho_{\text{rev}}^{\varphi} \leq s_\xi + N^\varphi_{\vee} s_{\text{c}}$. With~$s_\xi < -N^\varphi_{\vee} s_{\text{c}}$, we obtain~$\rho_{\text{rev}}^{\varphi}(\boldsymbol{x}) \leq 0$. From~\eqref{eq:sound}, the statement holds.
Although the formula in this example does not include multiple layers of~$\min$-type operators (e.g., $\max$-$\min$-$\max$-$\min$...), a similar proof applies to such cases due to the monotonicity of~$\min$, similar to the proof in Theorem~\ref{theo:satisfy}.
\end{proof}

\begin{remark}
Although TWP-CCP is technically a straightforward extension through penalization, weighting constraint violations based on robustness importance is enabled solely by the proposed decomposition approach, which provides a one-to-one connection between constraints in program~$\mathcal{P}_{\text{DC}}$ and the subtrees of~$\mathcal{T}^{\varphi}$, as stated in Proposition~\ref{prop:onetoone}.
\end{remark}

\section{Smooth Approximation by Mellowmin}\label{section:smooth}

One of the advantages of resulting DC program~$\mathcal{P}_{\text{DC}}$~\eqref{eq:final} is that smoothing only a small number of~$\min$ functions makes the program differentiable (when the functions in the assumptions are all differentiable). Although we can solve the non-differentiable program~$\mathcal{P}_{\text{DC}}$ directly using the subgradient of true~$\min$ function, this section explores the utilization of gradient-based algorithms. 

\subsection{Alternative smooth approximation}\label{subsec:smooth}

We propose a novel~$\min$'s smooth approximation suitable for our framework, which we call a mellowmin operator.
\begin{defi}\textit{(Mellowmin operator)}\label{defi:mellow}
\begin{equation}\label{minmellowdef}
\widetilde{\min}_k(a)=-\mathrm{mm}_k(-a),
\end{equation}
where~$a=(a_1,...,a_r)$,~$a_i \in\mathbb{R}$,~$k\in\mathbb{R}_{>0}$, and~$\mathrm{mm}_k(a)=\frac{1}{k}\ln \left(\frac{1}{r} \sum_{i=1}^r e^{k a_i}\right)$. 
\end{defi}
The mellowmax operator~$\mathrm{mm}_k(a)$ above was proposed as an alternative softmax operator in reinforcement learning in~\cite{Asadi_littman_2019, deepmellow2019p379} and can be considered as a log-\textit{average}-exp function, which is an under-approximation of~$\max$ function.

\begin{lemma}\label{lem:mellowmin_over} 
The over-approximation error bound of~$\widetilde{\min}_k$~\eqref{minmellowdef} is given by~$0 \leq \widetilde{\min}_k(a)-\min(a) \leq \frac{\log(r)}{k}$.
\end{lemma}
\begin{proof}
See Appendix~\ref{app:overapprox}.
\end{proof}

Here, we define a new robustness measure using the mellowmin function~$\widetilde{\min}_k$. The proposed robustness measure considers smooth approximation \textit{only} for the~$\min$ functions as~$\max$ functions are rather \textit{decomposed}.
\begin{defi}\textit{(Mellowmin robustness)}\label{def:newrobustness}
Given a formula~$\varphi$ and a trajectory~$\boldsymbol{x}$, the mellowmin robustness~$\widetilde{\rho}^{\varphi,k}(\boldsymbol{x})$ is defined as~$\widetilde{\rho}^{\varphi,k}(\boldsymbol{x}) = - \widetilde{\rho}^{\varphi,k}_{\text{rev}}(\boldsymbol{x})$, where~$k \in \mathbb{R}_{>0}$, and~$\widetilde{\rho}^{\varphi,k}_{\text{rev}}(\boldsymbol{x})$ denotes the reversed robustness function~$\rho^{\varphi}_{\text{rev}}(\boldsymbol{x})$ from \eqref{eq:robustnessrev}, with every~$\min$ operator replaced by the mellowmin operator~$\widetilde{\min}_k$ defined in \eqref{minmellowdef}. 
\end{defi}

\begin{theo}\textit{(Soundness)}\label{theo:soundness}
The mellowmin robustness~$\widetilde{\rho}^{\varphi,k}(\boldsymbol{x})$ is sound for $k\in\mathbb{R}_{>0}$, that is, $\widetilde{\rho}^{\varphi,k}_{\text{rev}}(\boldsymbol{x}) \leq 0\implies\boldsymbol{x} \vDash \varphi$.
\end{theo}
\begin{proof}
From Lemma~\ref{lem:mellowmin_over}, for any STL formula~$\varphi$ in negation normal form (NNF),
\begin{itemize}
  \item~$\widetilde{\rho}_{\text{rev}}^\mu(\boldsymbol{x}, t)=\rho_{\text{rev}}^\mu(\boldsymbol{x}, t)$
 \item~$\widetilde{\rho}_{\text{rev}}^{\varphi_1 \wedge \varphi_2}(\boldsymbol{x}, t) = \rho_{\text{rev}}^{\varphi_1 \wedge \varphi_2}(\boldsymbol{x}, t)$
 \item~$\widetilde{\rho}_{\text{rev}}^{\varphi_1 \vee {\varphi_2}}(\boldsymbol{x}, t) \geq \rho_{\text{rev}}^{\varphi_1 \vee \varphi_2}(\boldsymbol{x}, t)$
\item~$\widetilde{\rho}_{\text{rev}}^{ \square_{\left[t_1, t_2\right]} \varphi}(\boldsymbol{x}, t) = \rho_{\text{rev}}^{ \square_{\left[t_1, t_2\right]} \varphi}(\boldsymbol{x}, t)$
\item~$\widetilde{\rho}_{\text{rev}}^{ \Diamond_{\left[t_1, t_2\right]} \varphi}(\boldsymbol{x}, t) \geq \rho_{\text{rev}}^{ \Diamond_{\left[t_1, t_2\right]} \varphi}(\boldsymbol{x}, t)$
 \item~$\widetilde{\rho}_{\text{rev}}^{\varphi_1 \mathbf{U}_{\left[t_1, t_2\right]} \varphi_2}(\boldsymbol{x}, t) \geq \rho_{\text{rev}}^{\varphi_1 \mathbf{U}_{\left[t_1, t_2\right]} \varphi_2}(\boldsymbol{x}, t)$.
\end{itemize}
By induction, these inequalities follow that~$\widetilde{\rho}^{\varphi,k}_{\text{rev}}(\boldsymbol{x}) \leq 0\implies \rho^{\varphi}_{\text{rev}}(\boldsymbol{x})\leq 0$.
From~\eqref{eq:sound}, the statement follows.
\end{proof}

\begin{theo}\textit{(Asymptotic completeness)}\label{theo:complete} There exists a constant~$k_{\min}\in\mathbb{R}_{>0}$ such that 
\begin{equation}
\left| \rho^{\varphi}_{\text{rev}}(\boldsymbol{x})- \widetilde{\rho}^{\varphi,k}_{\text{rev}}(\boldsymbol{x}) \right| \leq \epsilon \text{ for all~$k \geq k_{\min}$}. \nonumber
\end{equation}
\end{theo}
\begin{proof}
See Appendix~\ref{app:theo_complete}.
\end{proof}

This theorem ensures that the mellowmin robustness~$\widetilde{\rho}^{\varphi,k}$ approaches the robustness~$\rho^\varphi$ as~$k$ approaches~$\infty$. In other words, the mellowmin robustness~$\widetilde{\rho}^{\varphi,k}$ asymptotically recovers the completeness property, that is,~$\widetilde{\rho}_{\text{rev}}^{\varphi,k}(\boldsymbol{x}) \geq 0 \implies   \boldsymbol{x} \nvDash \varphi$. Let us denote the program~\eqref{eq:moto} whose robustness functions~$\rho^{\varphi}_{\text{rev}}$ are replaced with~$\widetilde{\rho}^{\varphi,k}_{\text{rev}}$ as~$\widetilde{\mathcal{P}}$, and~$\widetilde{\mathcal{P}}$ after applying Algorithm~\ref{alg:reformulation} as~$\widetilde{\mathcal{P}}_{\text{DC}}$. This~$\widetilde{\mathcal{P}}_{\text{DC}}$ is the program~\eqref{eq:final} whose~$\min$ functions in~\eqref{eq:finalmin} are replaced by~$\widetilde{\min}_k$ operators.

The following corollary of Theorems~\ref{theo:satisfy} and~\ref{theo:sameoptimal} proves the equivalence and guarantees formula satisfaction for the reformulated program in the smoothed case.
\begin{corollary}
$\widetilde{\mathcal{P}}_{\text{DC}}$ is equivalent with~$\widetilde{\mathcal{P}}$. Moreover, a feasible solution~$\boldsymbol{z}$ of the program~$\widetilde{\mathcal{P}}_{\text{DC}}$ always satisfies $\varphi$, i.e.,~$\boldsymbol{x} \vDash \varphi$.
\end{corollary}
\begin{proof}
The proofs of Theorems~\ref{theo:satisfy} and~\ref{theo:sameoptimal} do not rely on specific assumptions about the robustness function, apart from the monotonicity property. Since the mellowmin function is strictly monotone with respect to each of its arguments (for any~\(k > 0\), note that~\(\frac{\partial \widetilde{\min}_k(a)}{\partial a_i} = -\frac{e^{-k a_i}}{\sum_{i=1}^r e^{-k a_i}} < 0\)), the inequality~\(\widetilde{\rho}_{\text{rev}}^{\varphi}(\boldsymbol{x}) \leq 0\) is satisfied. Therefore, the program~\(\widetilde{\mathcal{P}}_{\text{DC}}\) is equivalent to~\(\widetilde{\mathcal{P}}\). Using this inequality along with Theorem~\ref{theo:soundness}, the latter claim also holds.
\end{proof}
Through this statement, the three properties of Standard CCP discussed in Section~\ref{subsec:subproblem} hold also within the smoothed case.\footnote{On the other hand,~$\widetilde{\mathcal{P}}_{\text{DC}}$ may fail to find a solution even when~$\mathcal{P}_{\text{DC}}$ has one. In such cases, one can either increase the value of~$k$ or try different initial guesses.} 

Moreover, the strict monotonicity of the mellowmin function can aid the optimization process in finding a better solution.
For instance, consider the mellowmin version of~\eqref{eq:mellowmin}:~$\widetilde{\min}_k(\rho_{\text{rev}}^{\Psi_1}, \ldots, \rho_{\text{rev}}^{\Psi_{y_{\min}}}) \leq s_{\min}$. 
When any of the arguments of the mellowmin robustness function decreases in the iterative optimization process, the value on the left-hand side must also decrease as the mellowmin function is not solely determined by the most critical value. This behavior contrasts with the traditional robustness function, which is non-strictly monotone and can remain constant. Consequently, the bound~$s_{\min}$ on the right-hand side is effectively lowered, which could ultimately reduce the objective function value. 

\begin{figure*}[t]
 \begin{minipage}[b]{0.243\linewidth}
  \centering
 \mbox{\raisebox{0mm}{\includegraphics[keepaspectratio, scale=0.3]{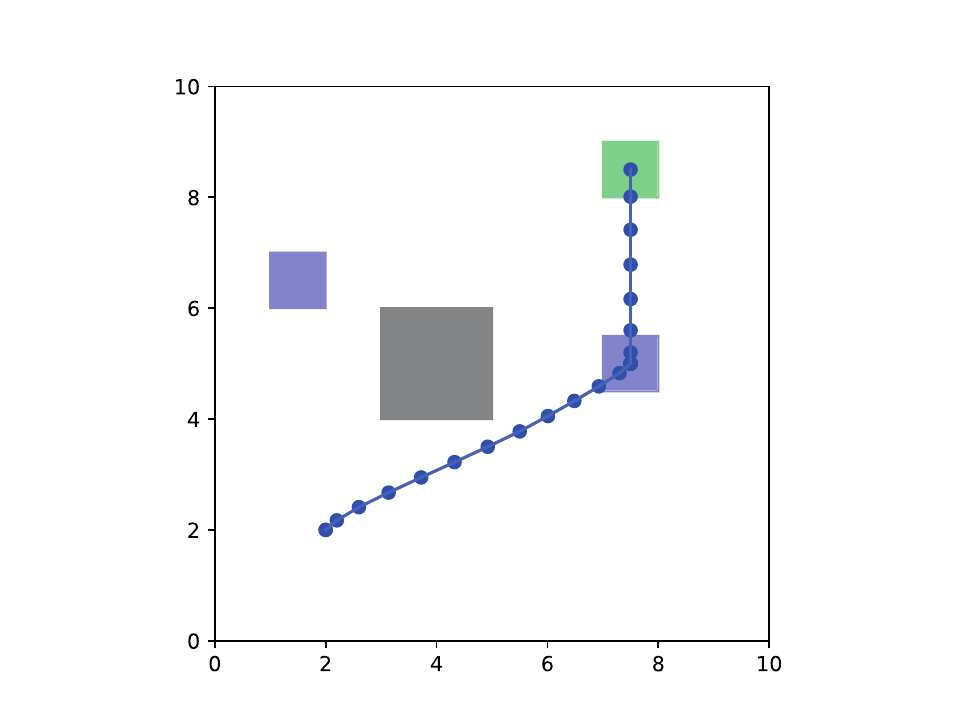}}}
  \subcaption{two-target}\label{either-MIP}
 \end{minipage}
 \begin{minipage}[b]{0.243\linewidth}
  \centering
  \includegraphics[keepaspectratio, scale=0.3]{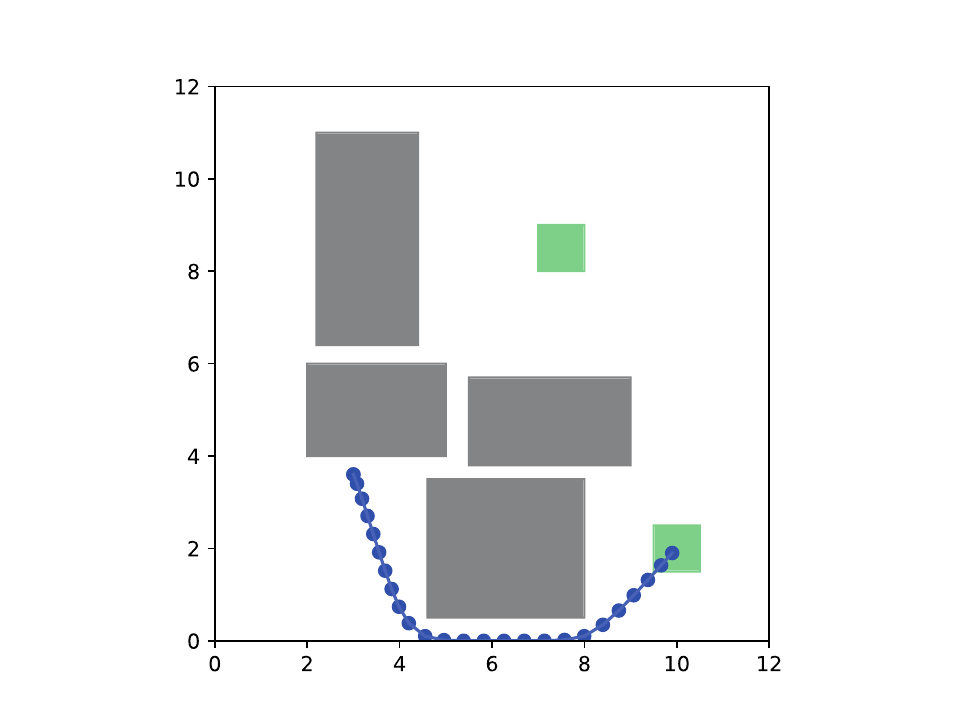}
  \subcaption{narrow-passage}\label{narrow-MIP}
 \end{minipage}
  \begin{minipage}[b]{0.243\linewidth}
  \centering
  \includegraphics[keepaspectratio, scale=0.3]{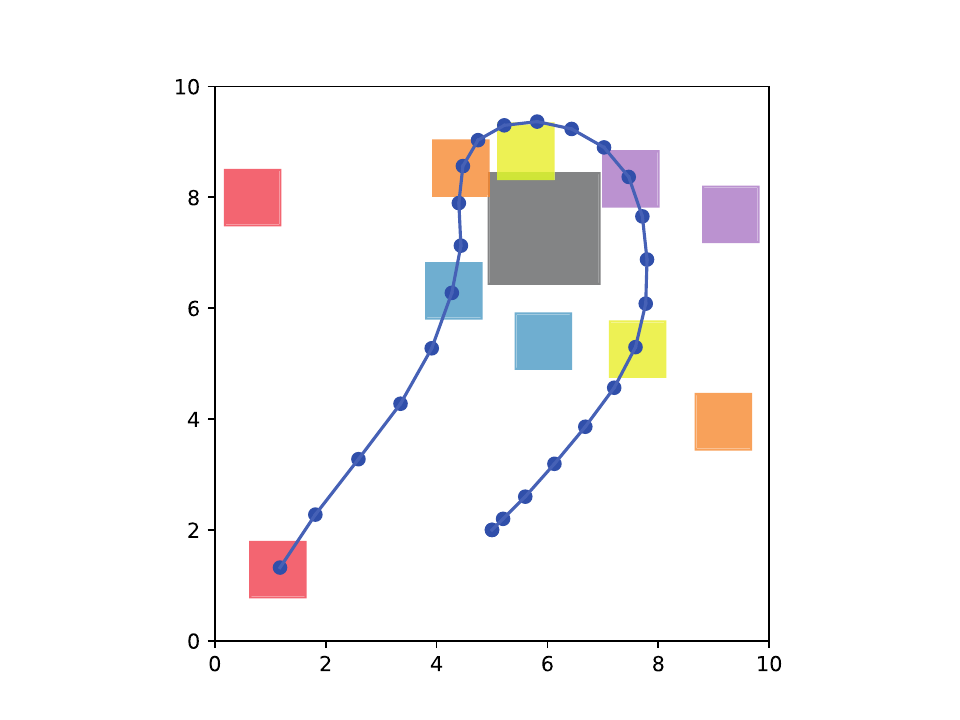}
  \subcaption{many-target}\label{multi-MIP}
 \end{minipage}
   \begin{minipage}[b]{0.243\linewidth}
  \centering
  \includegraphics[keepaspectratio, scale=0.3]{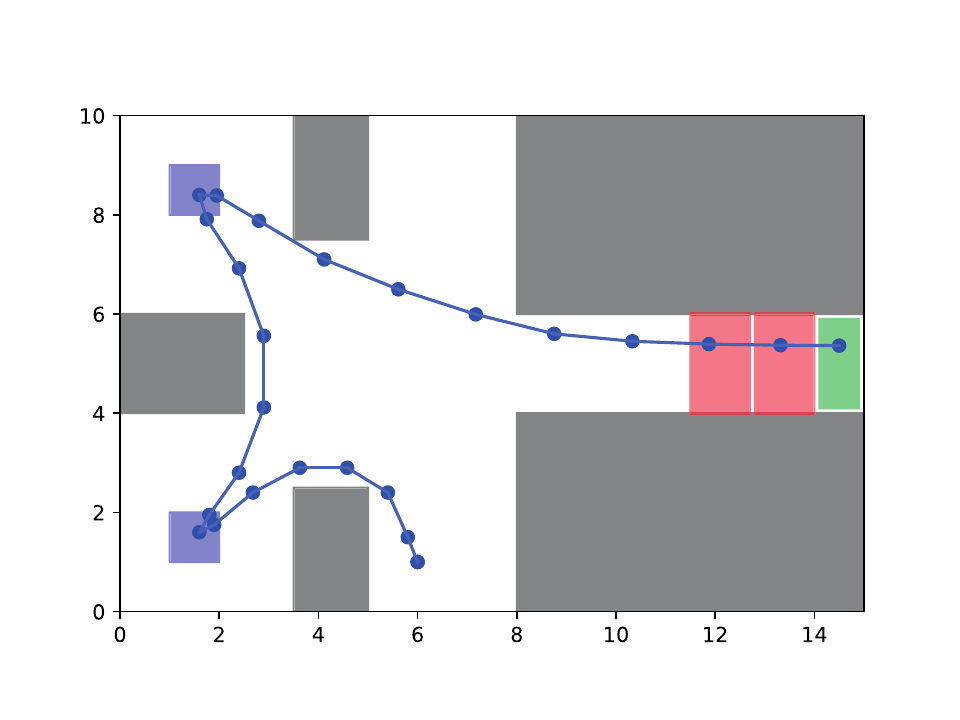}
  \subcaption{door-puzzle}\label{door-MIP}
 \end{minipage}
    \caption{Illustrations of scenarios along with solutions generated by the MIP-based method with~$H=25$. The greyed regions express obstacles ($O$) that the robot must avoid for all time steps. The green regions express goals ($G$) that the robot must reach at the last time step. Other regions colored purple, blue, orange, and pink represent areas that the robot must pass through (or stop by) at least one time step. Specifically, the STL specification for each scenario is as follows: \textbf{Two-target:}~$\Diamond_{[0, H-5]}(\square_{[0,5]} T_1 \vee \square_{[0,5]} T_2) \wedge \square_{[0, H]} \neg O \wedge \Diamond_{[0, H]} G$, \textbf{Narrow-passage:}~$\Diamond_{[0, H]}(G_1 \vee G_2) \wedge \square_{[0, H]}(\bigwedge_{i=1}^4 \neg O_i)$, \textbf{Many-target:}~$\bigwedge_{i=1}^5(\bigvee_{j=1}^2 \Diamond_{[0, H]} T_i^j) \wedge \square_{[0, H]}(\neg O)$, \textbf{Door-puzzle:}~$\bigwedge_{i=1}^2(\neg D_i \boldsymbol{U}_{[0, H]} K_i) \wedge \Diamond_{[0, H]} G \wedge \square_{[0, H]}(\bigwedge_{i=1}^5 \neg O_i)$.
    In the two-target and narrow-passage specifications, the robot must pass through one of the two same-colored regions. In the many-target specification, the robot must reach both same-colored regions. In the door-puzzle specification, the robot has to collect keys ($K$) in the two blue regions to open the corresponding doors represented by the red and green regions.}
    \label{fig;scenarioil}
 \normalsize
\end{figure*}

\subsection{Advantages of the mellowmin function}\label{subsec:other_smooth}
The most common smooth function for~$\max$ and~$\min$ functions in the robustness function~\eqref{eq:robustnessrev} is the log-sum-exp (LSE) function~\cite{Pant_smooth,Hashimoto2022-xu}:~$\overline{\max}_k:=\frac{1}{k} \ln \sum_{i=1}^r e^{k a_i}$ and~$\overline{\min}_k:=-\overline{\max }_k \left(-a\right)$. Although this smooth function is differentiable everywhere, the resulting robustness is not sound~\eqref{eq:sound}~\cite{Pant_smooth}. This is a considerable defect as a measure for guaranteeing safety. However, the mellowmin robustness function is sound and asymptotically complete as in Theorems~\ref{theo:soundness} and~\ref{theo:complete}. Therefore, the mellowmin function is preferred over the LSE function. 

The other popular smooth function is the soft-min operator~$\widehat{\min}_k:=-\frac{\sum_{i=1}^r a_i e^{-k a_i}}{\sum_{i=1}^r e^{-k a_i}}$, which is sound~\cite{Gilpin2021-wv}.
However, the function~$\widehat{\min}_k$ is unsuitable for the CCP-based approach due to its undetermined curvature  (convex or concave), while the mellowmin is concave for~$k > 0$, making it compatible with the CCP-based approach.
\begin{prop}\label{prop:mellowconcave}
The mellowmin~$\widetilde{\min}_k$ is concave for~$k>0$.
\end{prop}
\begin{proof}
The Hessian of the mellowmin is~
$$
\nabla^2 \widetilde{\min}_k (a) = -\nabla^2 \widetilde{\max}_k (-a) = -\nabla^2 \overline{\max}_k (-a)  
=-\frac{k}{\left(1_r^{\mathsf{T}}(-\sigma)\right)^2}\left(\left(1_r^{\mathsf{T}} (-\sigma)\right) \operatorname{diag}(-\sigma)-(-\sigma)(-\sigma)^{\mathsf{T}}\right),
$$
where $1_r$ is the $r$-dimensional vector of all ones, and $\sigma=(e^{k a_1},...,e^{k a_r})$. As~$\nabla^2 \widetilde{\max}_k (a)\geq 0$ for $k \geq 0$ \cite{deepmellow2019p379}, we have~$\nabla^2 \widetilde{\min}_k (a) \leq 0$ for $k \geq 0$. Thus, the proposition holds.
\end{proof}


Moreover, the concavity of the program~$\widetilde{\mathcal{P}}_{\text{DC}}$ increases \textit{monotonically} as~$k\to \infty$, approaching the original program~\eqref{eq:moto}, while it becomes a \textit{convex} program as~$k \to 0$.\footnote{Replacing all~$\min$ functions with average functions makes robustness a convex function. \cite{Lindemann2019-os} exploited this to formulate a convex quadratic program but had to significantly restrict STL specifications. In this sense, our work extends theirs by allowing all STL specifications.}
The degree of concavity of the program~$\widetilde{\mathcal{P}}_\text{DC}$ is bounded under Assumption~\ref{assum:linear}, as all concave constraints take the form of the mellowmin version of~\eqref{eq:finalmin}. The concavity of these constraints is bounded as follows:
\begin{prop}\label{prop:mellowbound}
The concaveness of the mellowmin operator is bounded by the following expression:
\begin{equation}\label{eq:lsebound}
-k \|v\|^2 \leq v^{\mathsf{T}} \nabla^2 \widetilde{\min}_k(a) v \leq 0, 
\end{equation}
where~$k\in \mathbb{R}_{>0}$, $v\in\mathbb{R}^{r}$, and~$\|v\|^2=v^\mathsf{T}v$. This bound becomes tighter \textit{monotonically} as~$k\rightarrow 0$. 
\end{prop}
The proof is omitted, as applying negative signs to the bound on $\nabla^2 \overline{\max}_k (= \nabla^2 \widetilde{\max}_k)$ in~\cite{Gao_softmax} directly yields this result. This property is appropriate for analyzing the tradeoff between smoothing accuracy and algorithm convergence because algorithm convergence typically improves with lower degrees of concavity (see~\cite{Debrouwere2013-ds_diehl,2011diehl}).

To summarize, among the possible smooth~$\min$ functions, including~$\overline{\min}_k$ and~$\widehat{\min}_k$, the mellowmin~$\widetilde{\min}_k$ is the only function that can preserve soundness (Theorem~\ref{theo:soundness}) while being concave (Proposition~\ref{prop:mellowconcave}). Furthermore, it allows for a trade-off between the precision of the smooth approximation (as~$k \rightarrow \infty$) and reduced concavity (as~$k \rightarrow 0$). For these reasons, the mellowmin function is better suited for CCP than the other two alternatives.


\section{Numerical Simulations}\label{section:example}
This section demonstrates the effectiveness of the proposed method over two other methods through four benchmark scenarios. The implementation is available at \url{https://github.com/yotakayama/STLCCP}.
During the experiment in Section~\ref{subsec:sum_experiments}, we also propose a practical remedy, the warm-start approach, to effectively utilize the mellowmin functions. All experiments were conducted on a MacBook Air 2020 with an Apple M1 processor (Maximum CPU clock rate: 3.2 GHz) and 8GB of RAM. All robustness scores in this section refer to the scores of the traditional robustness~$\rho^\varphi$.

\subsection{Numerical setup} \label{subsec:comparedmethods}

The state and control input are defined as~\(x_t = [p_{x_t}, p_{y_t}, \dot{p}_{x_t}, \dot{p}_{y_t}]^\mathsf{T}\in \mathbb{R}^{4}\) and~\(u_t = [\ddot{p}_{x_t}, \ddot{p}_{y_t}]^\mathsf{T}\in \mathbb{R}^{2}\), where~\(p_{x_t}\) and~\(p_{y_t}\) denote the robot's horizontal and vertical positions. The system follows a double integrator, i.e.,~\(x_{t+1} = A x_t + B u_t\) with  
$
A = \begin{bmatrix} I_2 & I_2 \\ 0 & I_2 \end{bmatrix}, \quad  
B = \begin{bmatrix} 0 \\ I_2 \end{bmatrix},
$
where~\(I_2\) and~\(0\) are the \(2 \times 2\) identity and zero matrices, respectively.

All benchmark scenarios are borrowed from~\cite{Kurtz2022-pe}. The illustration of all four scenarios is described in Figure~\ref{fig;scenarioil}.
The value of the initial state~$x_0$ for each specification is fixed as~$[2.0,2.0,0,0]^{\mathsf{T}}$, 
$[3.0,3.6,0,0]^{\mathsf{T}}$, $[5.0,2.0,0,0]^{\mathsf{T}}$, $[6.0,1.0,0,0]^{\mathsf{T}}$ respectively. The compared methods are summarized below.

\begin{figure*}[t]
 \centering
 \begin{minipage}[b]{0.47\linewidth}
 \centering
 \includegraphics[keepaspectratio, scale=0.5]{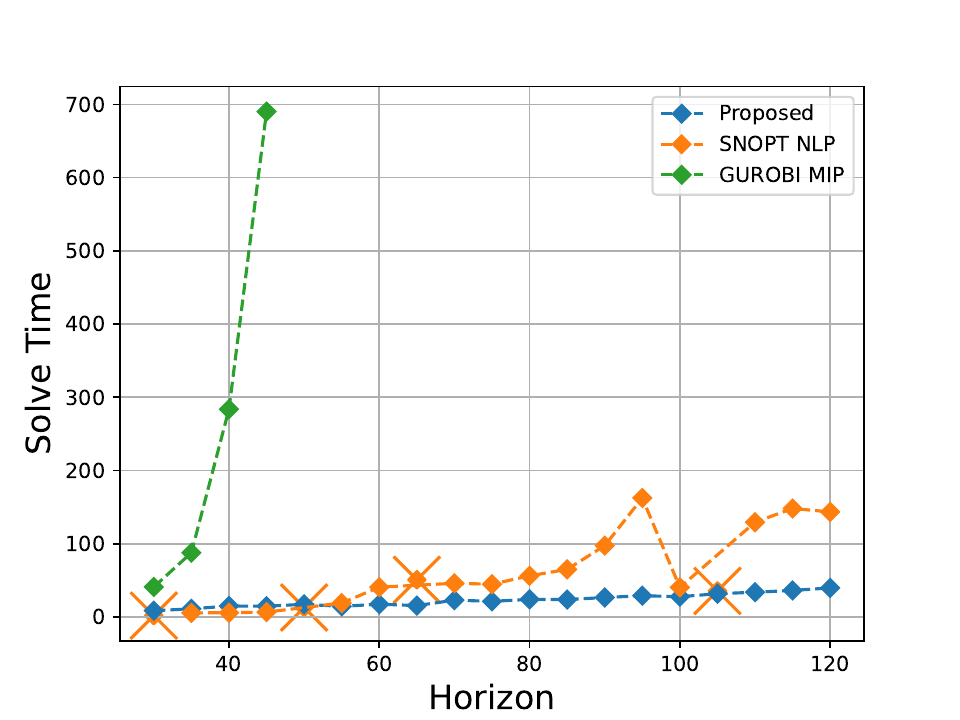}
 \subcaption{Computational time of the three methods.}\label{fig:solvetime_three}
 \end{minipage}
  \centering
 \begin{minipage}[b]{0.47\linewidth}
 \centering
 \mbox{\raisebox{0mm}{\includegraphics[keepaspectratio, scale=0.5]{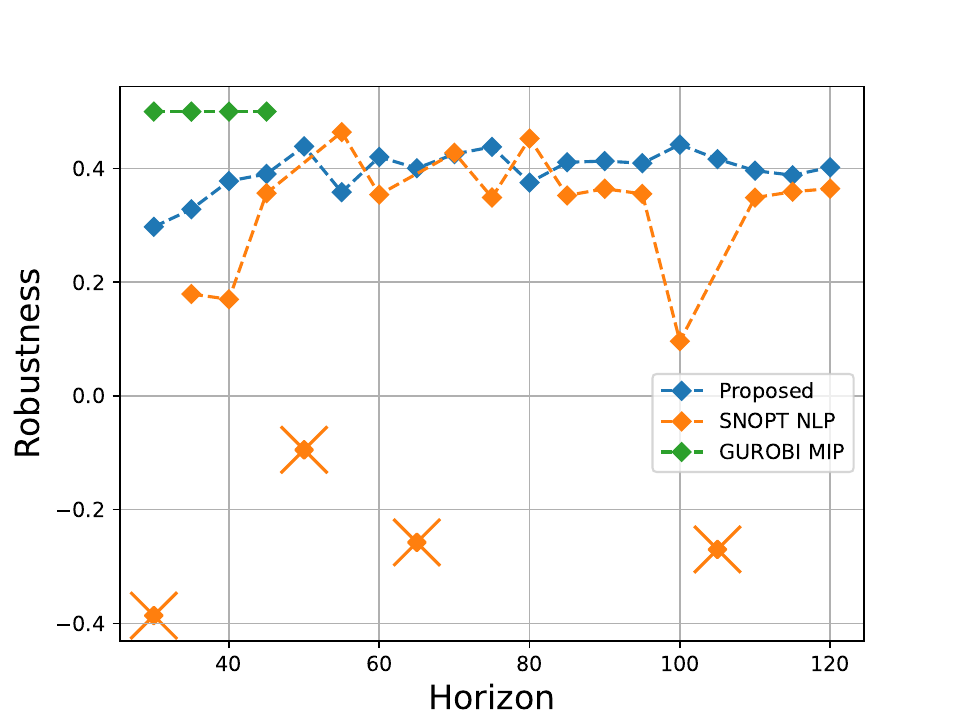}}}
 \subcaption{Robustness values of the three methods.}
 \label{fig:robustness_three}
 \end{minipage}
 \caption{Computational time and robustness score of the three methods in the many-target scenario over Horizon from~$H=30$ to~$120$.}\label{fig:three}
\end{figure*}

\textbf{Proposed approach:} 
We implemented the algorithm in Python using CVXPY~\cite{Diamond2016-ka} as the interface to the optimizer GUROBI (ver.10)~\cite{gurobi} (with the default parallel barrier algorithm).
For the robustness function, the first three experiments in Section~\ref{subsec:diff_horizons}--\ref{subsec:effectiveness_CCP} focus on the LSE smooth approximation~$\overline{\min}_k$ ($k=10$) to demonstrate the validity of the proposed reformulation method. The final subsection highlights the effectiveness of the mellowmin robustness function ($k=1000$). Regarding the optimization step, we use the TWP-CCP unless otherwise noted. To vary the initial guesses of the variables in each trial, we used a normal distribution.
The values of all CCP parameters described in Algorithm~\ref{alg:ccp} are listed in Table~\ref{tab:ccp_para} and remain fixed throughout the numerical experiments. 
\begin{table}[ht]
 \centering
  \caption{Parameter settings of CCP.}\label{tab:ccp_para}
 \begin{tabular}{|l|l|l|l|l|}
  \hline
   Parameters & Description & Value \\\hline \hline 
 ~$\tau_0$ &  Initial value of~$\tau_{(i)}$  & 5e-3 \\\hline  
  ~$\tau_{\max}$ & Maximum allowable value of~$\tau_{(i)}$ & 1e3 \\ \hline
 ~$\nu$ & Rate of increase for~$\tau_{(i)}$ & 2.0
  \\\hline 
 ~$s_{c}$ & Maximum tolerance for penalty variables  & 1e-5 \\\hline
 ~$s_{\text{ep}}$ & Maximum cost difference  & 1e-2 \\
  \hline 
 \end{tabular}
\end{table}

\textbf{MIP-based approach (GUROBI-MIP):} The problem is formulated as an MIP using the encoding framework in~\cite{Belta_undated-jj} and solved the problem with the GUROBI solver. Note that GUROBI is often the fastest MIP solver.

\textbf{SQP-based approach (SNOPT-NLP):} We employ a naive SQP approach with the SNOPT sparse SQP solver~\cite{GilMS05} (default parameter settings) and a sound smoothing method proposed in~\cite{Gilpin2021-wv} where LSE approximation is used for~$\min$ function and Boltzman softmax is used for~$\max$ function. Note that the SNOPT outperforms Scipy's SQP solver (similar arguments can be found in~\cite{Gilpin2021-wv, Kurtz2022-pe}).

\begin{remark}\textit{(Small quadratic cost)}
To demonstrate the effectiveness of the proposed framework in practical scenarios, we introduced a quadratic cost function~$w_q\sum_{t=0}^H (x_t^\mathsf{T} Q x_t+u_t^\mathsf{T} R u_t)$, where~$Q$ and~$R$ are positive semidefinite symmetric matrices and~$w_q$ is a weight parameter. The value of~$w_q$ was set to 0.01 or 0.001.
\end{remark}
\begin{remark}\textit{(A modification for the SNOPT-NLP)}\label{remark:snopt}
The SNOPT-NLP often finds Problem~\ref{problem} infeasible with the robustness constraint~\eqref{eq:min_xuu1}. To enable a fair comparison, we removed this constraint only for the SNOPT-NLP method, making the problem easier to solve.
\end{remark}

\subsection{Comparison over different horizons}\label{subsec:diff_horizons}
In the first experiment, we compared the proposed method with the SQP-based and MIP-based methods on the many-target specification, using time horizons from~$H=30$ to~$120$. The results are shown in Fig.~\ref{fig:three}, which includes computation times (Fig.~\ref{fig:solvetime_three}) and robustness scores (Fig.~\ref{fig:robustness_three}). Note that the SNOPT-NLP method occasionally shows negative robustness values because we omit the robustness constraint~\eqref{eq:min_xuu1} for this method (see Remark~\ref{remark:snopt}). The orange~$\color{orange} \times$ marker in Fig.~\ref{fig:solvetime_three} indicates SNOPT-NLP's failure at certain time horizons, meaning it gets stuck in local solutions with robustness less than~$0$. The blue plots represent the average values of five trials of the proposed method, as the results depend on initial variable values. We omit variance box plots for visibility (variance is shown in Fig.~\ref{fig:ccp}). The results of the other two methods do not change by trial.

The MIP-based method's computation time increases significantly with longer horizons. The proposed method finds satisfactory trajectories for all horizons, while the SQP-based method fails at~$H=30, 50, 65,$ and~$105$ (marked by orange~$\color{orange} \times$). The proposed method's robustness scores consistently exceed zero across all trials, succeeding in~$95$ out of~$95$ experiments. The SNOPT-NLP method succeeds in~$15$ out of~$19$ samples, even without the robustness constraint.

The proposed method's convergence time increases slightly with longer horizons, while the SNOPT-NLP method's performance fluctuates and takes longer to converge for~$H \geq 55$, often with poorer robustness. The SNOPT-NLP method's apparent decrease in computation time at~$H=100, 105$ is due to failures or low robustness values, indicating early optimization stagnation in a local optimum. The proposed method consistently achieves higher robustness scores, close to the global optimum of the MIP-based method. The success of the proposed method is due to the presence of few concave constraints~\eqref{eq:finalmin},~$81$, out of the total~$3579$ constraints in~\eqref{eq:final}.



\begin{figure*}[t]
 \begin{minipage}[b]{{0.243\linewidth}}
 \centering
 \mbox{\raisebox{0mm}{\includegraphics[keepaspectratio, scale=0.3]{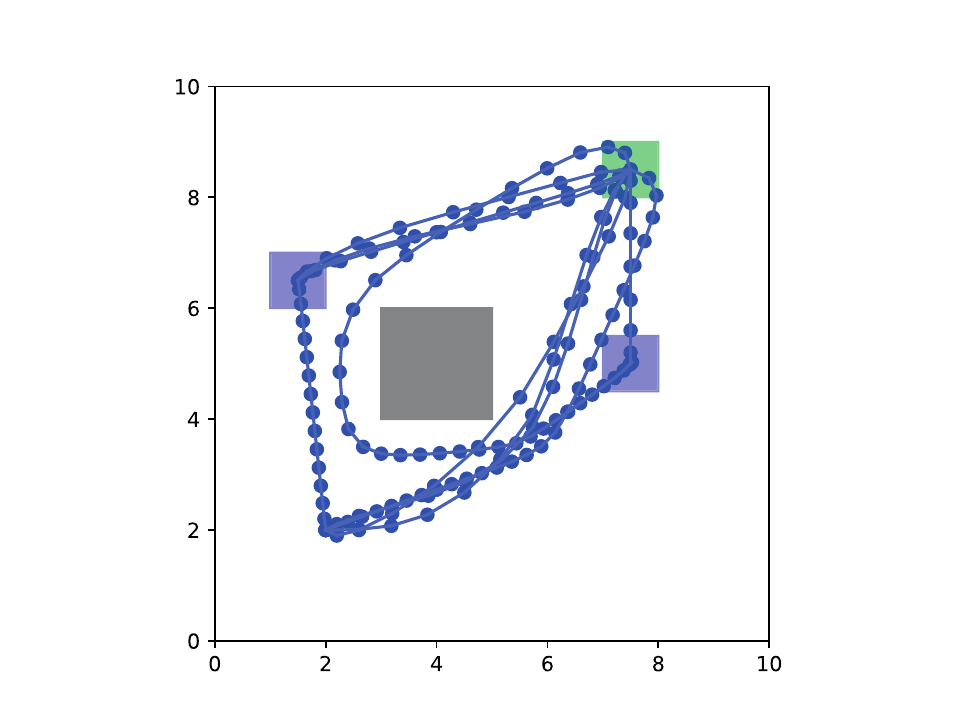}}}
 \subcaption{two-target}\label{either-NLP}
 \end{minipage}
 \begin{minipage}[b]{{0.243\linewidth}}
 \centering
 \includegraphics[keepaspectratio, scale=0.3]{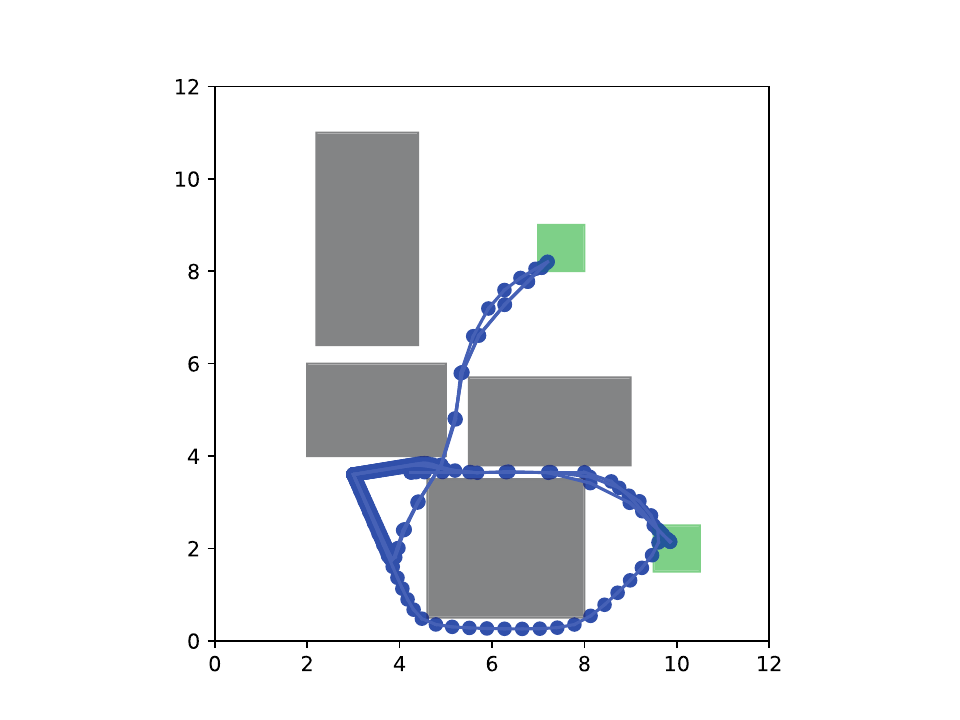}
 \subcaption{narrow-passage}\label{narrow-NLP}
 \end{minipage}
 \begin{minipage}[b]{{0.243\linewidth}}
 \centering
 \includegraphics[keepaspectratio, scale=0.3]{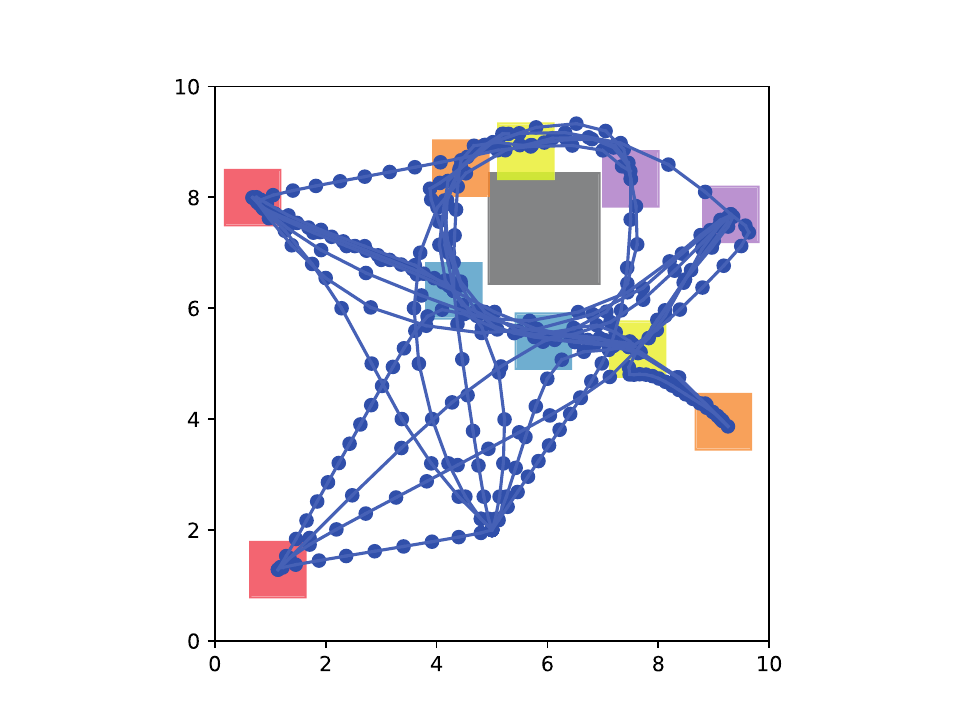}
 \subcaption{many-target}\label{multi-NLP}
 \end{minipage}
  \begin{minipage}[b]{{0.243\linewidth}}
 \centering
 \includegraphics[keepaspectratio, scale=0.3]{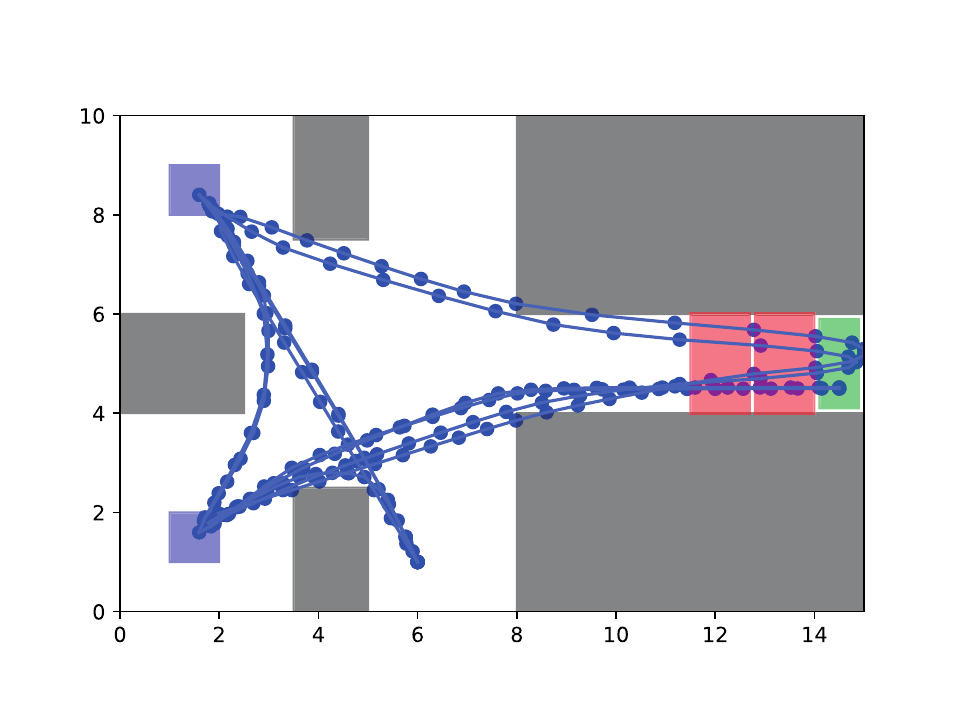}
 \subcaption{door-puzzle}\label{door_NLP}
 \end{minipage}
 \caption{Illustrations of scenarios along with the satisfactory trajectories generated by the proposed method with~$H=50$. The proposed method produced a range of satisfactory trajectories, each dependent on the initial guesses.}\label{fig:various_solutions}
 \normalsize
\end{figure*}

\begin{table*}[ht]
\tiny
 \caption{Solve times and robustness scores for the four different scenarios. Each value for the proposed method represents the average of successful trials out of 10. If the robustness score for the SNOPT-NLP method is negative, it is labeled as ``Failed." If the solve time for GUROBI-MIP exceeds the allotted time of 7500, it is labeled as ``Time out."}
\centering
\resizebox{\textwidth}{!}{
\begin{tabular}{|c|c||c|c|c||c|c|c||c|c|c|} 
\hline
Scenario & Horizon ($H$) & \multicolumn{3}{c||}{Solve Time (s)} & \multicolumn{3}{c||}{Robustness} & \multicolumn{3}{c|}{Success Rate} \\ \hline
~ & ~ & Ours & SNOPT-NLP & MIP & Ours & SNOPT-NLP & MIP & Ours & SNOPT-NLP & MIP \\
\hline \hline
 & 50 & 23.13 & \textbf{15.28} & 15.08 & \textbf{0.494} & 0.367 & 0.500 & 80.0 \% & 100.0 \% & 100.0 \%\\ 
two-target & 75 & \textbf{26.76} & Failed & 193.82 & \textbf{0.500} & Failed & 0.500 & 90.0 \% & Failed & 100.0 \%\\ 
 & 100 & \textbf{37.19} & 57.24 & 337.31 & \textbf{0.500} & 0.462 & 0.500 & 100.0 \% & 100.0 \% & 100.0 \%\\ \hline
 & 50 & \textbf{16.19} & Failed & 2526.99 & \textbf{0.440} & Failed & 0.500 & 100.0 \% & Failed & 100.0 \%\\ 
many-target & 75 & \textbf{21.45} & 44.78 &~$>$7500.00 & \textbf{0.430} & 0.349 & Time Out & 100.0 \% & 100.0 \% & Time Out\\ 
 & 100 & \textbf{27.80} & 40.39 &~$>$7500.00 & \textbf{0.437} & 0.096 & Time Out & 100.0 \% & 100.0 \% & Time Out\\ \hline
 & 50 & \textbf{22.21} & Failed & 36.74 & \textbf{0.151} & Failed & 0.400 & 100.0 \% & Failed & 100.0 \%\\ 
narrow-passage & 75 & 35.18 & \textbf{8.83} &~$>$7500.00 & 0.027 & \textbf{0.151} & Time Out & 80.0 \% & 100.0 \% & Time Out\\ 
 & 100 & 57.41 & \textbf{39.94} &~$>$7500.00 & 0.111 & \textbf{0.170} & Time Out & 100.0 \% & 100.0 \% & Time Out\\ \hline
 & 50 & \textbf{299.33} & Failed &~$>$7500.00 & \textbf{0.348} & Failed & Time Out & 30.0 \% & Failed & Time Out\\ 
door-puzzle & 75 & Failed & Failed &~$>$7500.00 & Failed & Failed & Time Out &~$<$ 5.0 \% & Failed & Time Out\\ 
 & 100 & Failed & Failed &~$>$7500.00 & Failed & Failed & Time Out &~$<$ 5.0 \% & Failed & Time Out\\ \hline
\end{tabular}
}
\label{fig:combined_table} 
\end{table*}

\subsection{Comparison over different specifications}\label{subsec:diff_spec}

The next experiment compared the three methods across four scenarios: two-target, narrow-passage, many-target, and door-puzzle. Table~\ref{fig:combined_table} shows the convergence times, robustness scores, and success rates. The proposed method consistently finds satisfactory trajectories for all horizons in the two-target, many-target, and narrow-passage scenarios, outperforming the others. It also succeeds in the door-puzzle scenario with~$H=50$. The proposed method's success rate is above 90\% in the two-target and narrow-passage scenarios and 100\% in the many-target scenario, showing robustness against initial variable guesses. In the door-puzzle scenario, the success rate is 30\%, which is still notable for these challenging scenarios. The proposed method also occasionally succeeded for~$H=75$ and~$H=100$, but we excluded these cases due to their low success rates.

The proposed method is the fastest in most horizons for the two-target and many-target scenarios, with slight increases in convergence time as the horizon lengthens. The SNOPT-NLP method occasionally performs better in the narrow-passage scenario but exhibits high volatility and random failures. 
The proposed method's robustness scores are consistently satisfactory (not only the averages of ten trials given different initial guesses), unlike the SNOPT-NLP method's volatile performance. Fig.~\ref{fig:various_solutions} shows final trajectories generated by the proposed method with different initial guesses for all scenarios with~$H=50$, demonstrating that the proposed method is robust against initial guesses.

The proposed method's relatively poorer performance in the narrow-passage scenario is likely due to the narrow path and problem relaxation by penalty variables. The CCP’s majorization targets only the goal specification~$\Diamond_{[0, H]}\left(G_1 \vee G_2\right)$ (with a disjunctive node that unifies~$\Diamond_{[0, H]}$ and~$\vee$), not obstacle avoidance~$\square_{[0, H]}\left(\bigwedge_{i=1}^4 \neg O_i\right)$. 
Although the relaxation of the goal specification itself is unlikely the primary issue (supported by the effectiveness of the proposed method in the many-target scenario, which involves even more disjunctive goal nodes), this relaxation, combined with the obstacle avoidance specification, slows down the convergence of the solution. For instance, a trajectory passing through the center of obstacles can be generated easily with the narrow feasible region because of the large penalty variables allowed in the early optimization steps. These violations are then gradually reduced by incrementally increasing the weight of penalty variables. Guided by the robustness function, the trajectory iteratively converges toward the narrow path (i.e., the feasible region). However, this process either significantly increases computational time or fails to reach the narrow path. Adjusting parameters could improve performance, for instance, by increasing the penalty variable weights to achieve faster convergence and tightening the termination conditions to improve robustness scores. Exploring these optimal parameter settings remains an interesting area for future research.

\subsection{Effectiveness of the TWP-CCP}\label{subsec:effectiveness_CCP}

In the third experiment, we compared three methods: standard CCP, TWP-CCP, and a variant called \textit{TWP-CCP with decay}. The decay variant starts as TWP-CCP and exponentially approaches standard CCP (see Appendix~\ref{app:detailed} for details). We evaluated these methods in the many-target scenario with~\mbox{$H=75$}, using 20 random initial guesses. The results, shown in Figs.~\ref{fig:box_solvetime} and~\ref{fig:box_robustness}, indicate that both TWP-CCP methods achieve consistently high robustness scores, independent of the initial guesses, and almost halve the convergence times compared to standard CCP on average.

This suggests that incorporating priority importance among tree nodes significantly enhances algorithm stability and reduces dependence on initial guesses. This was confirmed by analyzing variable changes in standard CCP during optimization. Failures in standard CCP often resulted from high values of variables associated with important constraints, leading to local optima. Additionally, the comparable performance of the decay method to TWP-CCP highlights the importance of incorporating priority ranking in the early stages of the optimization process. Significant violations of important constraints in the early stages, resulting from problem relaxation, can increase the likelihood of getting stuck in a local minimum outside the feasible region during subsequent iterations.

 \begin{figure*}[t]
 \begin{minipage}[b]{0.5\linewidth}
 \centering
 \includegraphics[keepaspectratio, scale=0.5]{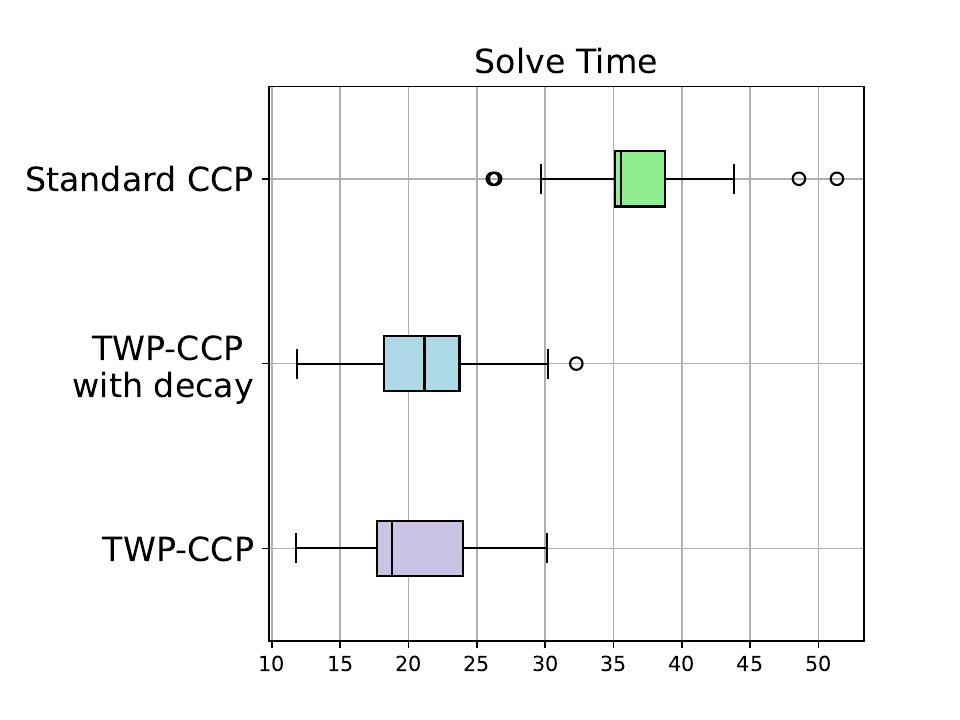}
 \subcaption{Computational time}\label{fig:box_solvetime}
 \end{minipage}
 \begin{minipage}[b]{0.5\linewidth}
 \centering
 \includegraphics[keepaspectratio, scale=0.5]{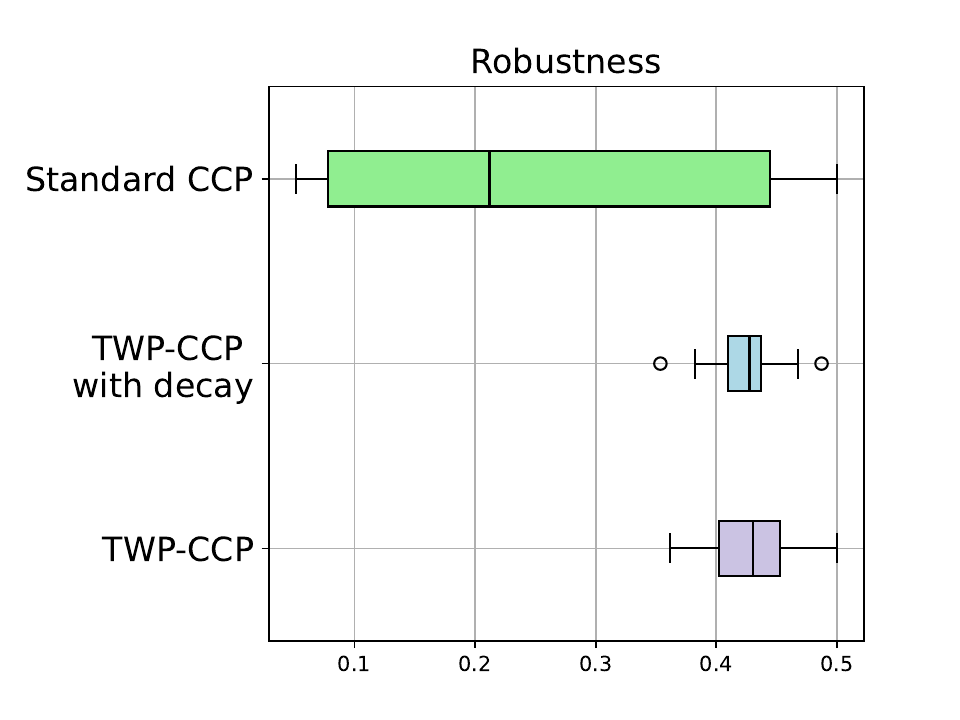}
 \subcaption{Robustness}\label{fig:box_robustness}
 \end{minipage}
 \caption{Box plots of 20 samples for each method for the many-target scenario with~$H=75$. }\label{fig:ccp}
\end{figure*}

\subsection{Effectiveness of the mellowmin smoothing}\label{subsec:sum_experiments}

Across all experiments above, we used LSE smoothing with~$k=10$ to demonstrate the validity of the proposed framework. However, LSE smoothing lacks desirable properties like soundness, which mellowmin smoothing provides. Solving the mellowmin smoothed program~$\widetilde{\mathcal{P}}_{\text{DC}}$ solely often results in poorer solutions because the smooth parameter~$k$ should be set high (e.g., 1000) to accurately approximate the true~$\min$ function, while~$k=10$ is sufficient for the LSE smoothing case. This high value can make the function more prone to being stuck in a local minimum as this high value of~$k$ may make the function more likely to be affected by only one argument with the critical value, similar to the true~$\min$ function (see Section~\ref{subsec:other_smooth}).
Therefore, we propose using the solution from LSE smoothing as a warm start for the mellowmin smoothed program. This approach achieves higher robustness scores while guaranteeing formula satisfaction.

Fig.~\ref{fig:repeat} compares LSE smoothing and mellowmin smoothing using the LSE solution as a warm start in the many-target specification. The additional time to solve the mellowmin smoothed program is shown in Fig.~\ref{fig:box_solvetime_repeat}. The average computational time remains low, approximately 10, even as the planning horizon increases. 
Despite the low computational time, the robustness score improves, occasionally achieving the global optimum of~$0.5$. This is likely because a solution from the LSE-smoothed program is often already a good trajectory. It is worth noting that the mellowmin smoothing approach always finds a solution with the fixed parameter~$k=1000$ without any modification.
We also validated this warm-start approach in the door-puzzle scenario. By repeating the warm-started mellowmin approach, we observed an increase in the success rate over the rate with the LSE smoothing approach in Section~\ref{subsec:diff_spec} (for further details, please refer to the data provided in the GitHub link).

 \begin{figure*}[!ht]
   \begin{minipage}[b]{0.5\linewidth}
 \centering
 \includegraphics[keepaspectratio, scale=0.5]{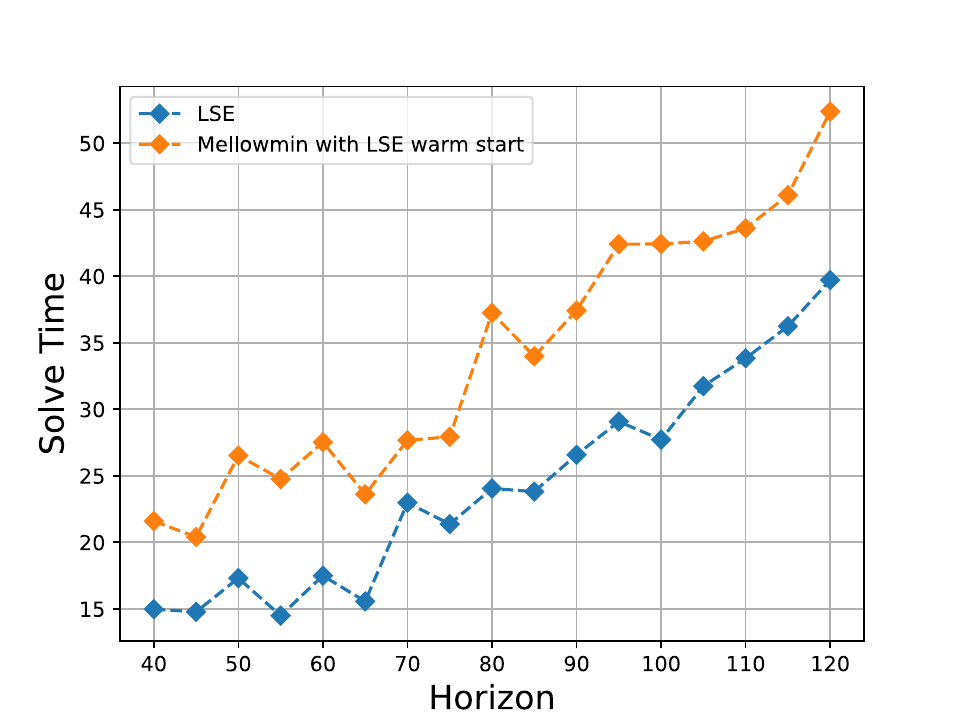}
 \subcaption{Solve Time}\label{fig:box_solvetime_repeat}
 \end{minipage}
 \begin{minipage}[b]{0.5\linewidth}
 \centering
 \includegraphics[keepaspectratio, scale=0.5]{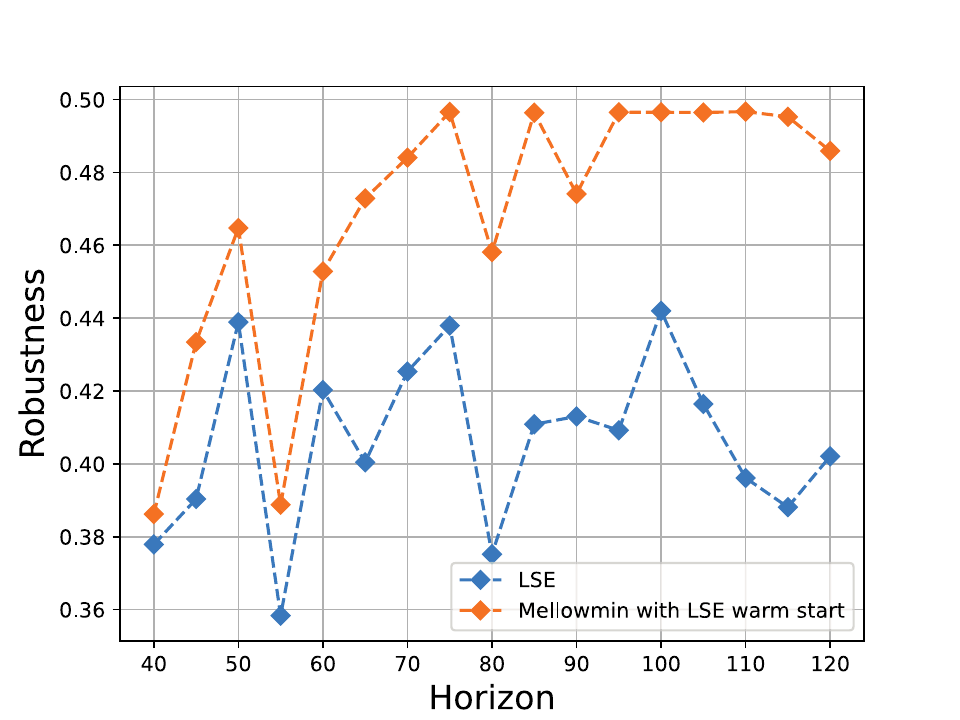}
 \subcaption{Robustness}\label{fig:box_robustness_repeat}
 \end{minipage}
 \caption{Computational time and robustness score of the two methods in the many-target scenario over Horizon from~$H=40$ to~$120$. The orange plot additionally solves the mellowmin smoothed program \protect{$\widetilde{\mathcal{P}}_{\text{DC}}$} using the solution of the blue plots.}\label{fig:repeat}
\end{figure*}

\section{Discussion and Conclusion}\label{section:conclusion}

\subsection{Discussion of key properties}\label{subsec:discuss}

The strength of the proposed framework lies in exploiting key properties of STL specifications. Subsequently, we elaborate on how these properties are utilized.

\textbf{Conjunctive-disjunctive to convex-concave:} We formally established a \textit{correspondence} between the conjunctiveness/disjunctiveness of the temporal operators and the convexity/concavity of the optimization program. Thus far, we first introduced the reversed version of the robustness function in Section~\ref{subsec:stl} to associate all the logical operators, described using \textit{conjunctive} or \textit{disjunctive} terms, with the~$\max$ or~$\min$ functions of the robustness function. This enabled us to reformulate the problem into a DC program by mapping conjunctive operators to convex parts of the program (and disjunctive operators to concave parts) in Section~\ref{section:decomposition} (in particular, the one-to-one correspondence in Proposition~\ref{prop:onetoone}).

\textbf{Monotonicity property:} The monotonicity property of the robustness function was required to prove Theorems~\ref{theo:satisfy} and~\ref{theo:sameoptimal} to make the robustness transformation equivalent. Additionally, the strict monotonicity of the mellowmin robustness function aided in finding a better optimal solution as in Section~\ref{subsec:smooth}.

\textbf{Hierarchical tree structure:} The TWP-CCP in Section~\ref{subsec:twp-ccp} used the hierarchical priority ranking of constraints. Unlike previous work, the proposed robustness decomposition enabled a more nuanced differentiation of node importance, resulting in a detailed method for prioritizing constraints.

Interpreting the robustness function as a tree was fundamental to several aspects of this paper. It facilitated the hierarchical prioritization described above, while also supporting key analyses. For example, the definition of the robustness tree structure in Section~\ref{subsec:tree_structure} was essential for characterizing the assumptions required for the proposed approach (e.g., Assumption~\ref{assum:dcpredicate}). Additionally, this tree perspective was critical to the proof of Theorem~\ref{theo:sameoptimal} and provided the basis for the tree simplification technique introduced in Section~\ref{subsection:simplification}, which minimizes concave constraints.

\begin{figure}[!t]
    \centering
    \includegraphics[keepaspectratio, scale=0.42]{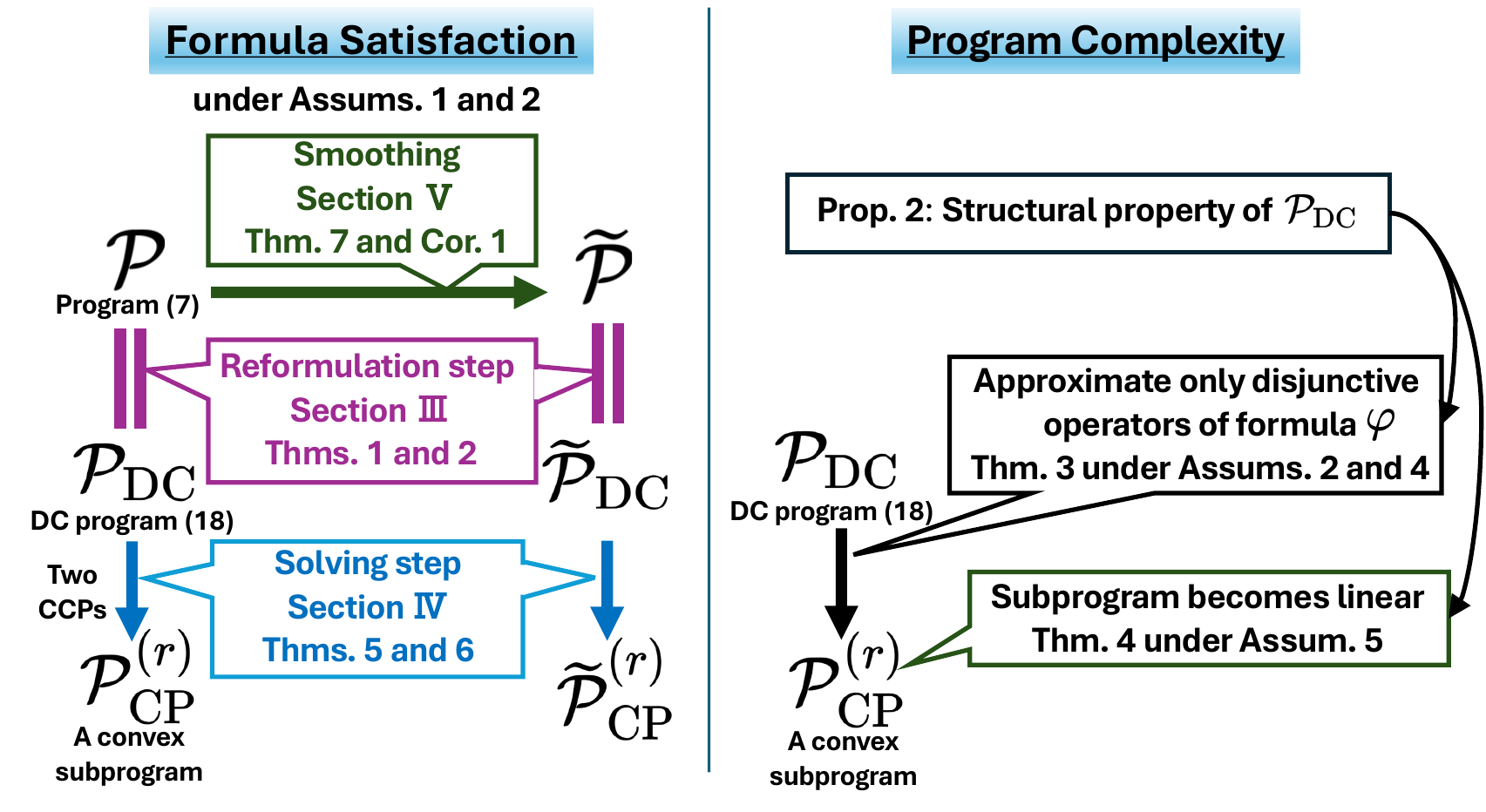}
    \caption{Summary of main results in this paper}
    \label{fig:overview}
\end{figure}

\subsection{Conclusion}\label{subsection:conclu}
This study established a connection between STL and CCP and leveraged this connection to propose an efficient optimization framework. 
An overview of the main results is provided in Fig.~\ref{fig:overview}. The left-hand side of the figure illustrates how the proposed structure-aware decomposition of STL formulas enables encoding control problems~$\mathcal{P}$ into a specific form of DC program~$\mathcal{P}_{\text{DC}}$. This reformulation is equivalent to the original problem and guarantees formula satisfaction, i.e., $\boldsymbol{x} \vDash \varphi$, under Assumptions~\ref{assum:dc}--\ref{assum:dcpredicate}. Furthermore, both convex subprograms~$\mathcal{P}_{\text{CP}}$ and~$\mathcal{P}_{\text{CP}}^r$, after undergoing majorization in the standard CCP and TWP-CCP settings, respectively, also guarantee formula satisfaction.
The right-hand side of Fig.~\ref{fig:overview} highlights the complexity analysis, showing that under specific assumptions (Assumptions~\ref{assum:major}--\ref{assum:linear}), subprogram~$\mathcal{P}_{\text{CP}}$ simplifies to a linear program by approximating only~$N_\vee^\varphi$ concave constraints related to the disjunctive nodes in the STL tree.
To enhance the efficiency of the algorithm with gradient-based solvers, we introduced the mellowmin-based approximation~$\widetilde{\mathcal{P}}$, which is also compatible with the theorems established above. In collaboration, these contributions form the STLCCP framework, which demonstrated scalability and robustness in synthesizing control strategies for motion planning problems. Future work will focus on extending the framework to various applications beyond controller synthesis, including problems involving time delays and real-time uncertainties.

\bibliographystyle{IEEEtran}

\newpage

\section{Appendix}
\subsection{Proof of Lemma~\ref{lem:mellowmin_over}}\label{app:overapprox}
The difference between the mellowmax function~$\mathrm{mm}_k$ and the~$\max$ function is given by
\begin{equation}
\mathrm{mm}_k(a) - \max(a) = \frac{-\log(r) + \log\left( \sum_{i=1}^r e^{k \left( a_i - \max(a) \right)} \right)}{k}, \nonumber \label{eq:mm_max_eq}
\end{equation}
where~$a = (a_1, \dots, a_r)$, $a_i\in\mathbb{R}$, and~$k \in \mathbb{R}_{>0}$. We define the number of arguments equal to the maximum value among the arguments as $r_m := \left| \{ i \in \{1, \dots, r\} : a_i = \max(a) \} \right|$, where~$1 \leq r_m \leq r$. The term~$\log\left( \sum_{i=1}^r e^{k \left( a_i - \max(a) \right)} \right)$ attains its maximum value~$\log(r)$ when each argument of vector~$a$ is equal, i.e., when~$r_m = r$. This results in the inequality
$\mathrm{mm}_k(a) - \max(a) \leq 0$.
Conversely, when~$r_m \neq r$,~$k \left( a_i - \max(a) \right) \leq 0$ for~$i \in \{1, \dots, r\}$ and at least one exponent is strictly less than 0. Since the exponential function is always positive, we have~$\mathrm{mm}_k(a) - \max(a) \geq -\frac{\log(r)}{k}$. By taking the negative of the entire expression and negating each argument, we get~$0 \leq -\mathrm{mm}_k(-a) + \max(-a) \leq \frac{\log(r)}{k}$.
Hence, the mellowmin function is an over-approximation of the true~$\min$, and the statement holds.

\subsection{Proof of Theorem~\ref{theo:complete}} \label{app:theo_complete}
Recall from Lemma~\ref{lem:mellowmin_over} that for~$k \geq \frac{\log (r)}{\epsilon}$, we have
\begin{equation}\label{eq;minimini}
|\min (a)-\widetilde{\min }(a)| \leq \epsilon .
\end{equation}
By applying \eqref{eq;minimini} to the definition of~$\widetilde{\rho}_{\text{rev}}^{\varphi,k}$~\eqref{eq:robustnessrev}, we have 
\begin{itemize}
  \item~$\left|\widetilde{\rho}_{\text{rev}}^\mu(\boldsymbol{x}, t)-\rho_{\text{rev}}^\mu(\boldsymbol{x}, t)\right|=0$
\item~$\left|\widetilde{\rho}_{\text{rev}}^{\varphi_1 \wedge \varphi_2}(\boldsymbol{x}, t)-\rho_{\text{rev}}^{\varphi_1 \wedge \varphi_2}(\boldsymbol{x}, t)\right| \leq 0$ 
\item~$\left|\widetilde{\rho}_{\text{rev}}^{\varphi_1 \vee \varphi_2}(\boldsymbol{x}, t)-\rho_{\text{rev}}^{\varphi_1 \vee \varphi_2}(\boldsymbol{x}, t)\right| \leq \epsilon$
\item~$\left|\widetilde{\rho}_{\text{rev}}^{ \square_{\left[t_1, t_2\right]} \varphi}(\boldsymbol{x}, t)-\rho_{\text{rev}}^{ \square_{\left[t_1, t_2\right]} \varphi}(\boldsymbol{x}, t)\right| \leq 0$
\item~$\left|\widetilde{\rho}_{\text{rev}}^{ \Diamond_{\left[t_1, t_2\right]} \varphi}(\boldsymbol{x}, t)-\rho_{\text{rev}}^{ \Diamond_{\left[t_1, t_2\right]} \varphi}(\boldsymbol{x}, t)\right| \leq \epsilon$
\item~$\left|\widetilde{\rho}_{\text{rev}}^{\varphi_1 \boldsymbol{U}_{\left[t_1, t_2\right]} \varphi_2}(\boldsymbol{x}, t)-\rho_{\text{rev}}^{\varphi_1\boldsymbol{U}_{\left[t_1, t_2\right]} \varphi_2}(\boldsymbol{x}, t)\right| \leq \epsilon$.
\end{itemize}
By induction, the statement follows.

\subsection{Parameter settings in Section~\ref{subsec:effectiveness_CCP}}\label{app:detailed}
\textbf{TWP-CCP with decay}: We adopt the TWP-CCP in Algorithm~\ref{alg:ccp}, using a modified cost term:
\begin{equation}\label{eq:expdecay}
    \tau_{(i)} \sum_{j=1}^{N^\varphi_{\vee}} 
    \Big(N_p^{\varphi_j} 
        - \min_{j}(N_p^{\varphi_j}) \exp(-e(i-1)) 
        + \min_{j}(N_p^{\varphi_j}) 
    \Big) s_j,\nonumber
\end{equation}
where~$i$ denotes the iteration step, and~$e$ is a parameter that controls the rate of exponential decay, set to~$0.2$ in this case. The term~$\min_{j}\big(N_p^{\varphi_j}\big)$ denotes the smallest number of leaves among the subtrees~$\mathcal{T}^{\varphi_j}$ of concave constraints (e.g., 4 in the multi-target scenario with~$H = 75$).

\textbf{Standard CCP}: We adopt the penalty CCP~\cite{Lipp2016-fa}, using the cost term:
$
\tau_{(i)} \min_{j}(N_p^{\varphi_j}) \sum_{j=1}^{N^\varphi_{\vee}} s_j,\nonumber
$
where the penalty weight is fixed to~$\min_{j}(N_p^{\varphi_j})$ for all $N^\varphi_{\vee}$ concave constraints. 

\end{document}